%% file: main.tex
\newif\ifcameraready
  \camerareadyfalse

\documentclass[a4paper, UKenglish, numberwithinsect, cleveref, autoref, thm-restate, pdfa]{lipics-v2021}
\usepackage{mathtools}
\usepackage{bbding}

\title{Responsibility Measures for Conjunctive Queries with Negation}

\author{Meghyn Bienvenu}{Univ. Bordeaux, CNRS, Bordeaux INP, LaBRI, UMR5800, F-33400 Talence, France\and
\url{https://www.labri.fr/perso/meghyn/}}{meghyn.bienvenu@cnrs.fr}{https://orcid.org/0000-0001-6229-8103}{}
\author{Diego Figueira}{Univ. Bordeaux, CNRS, Bordeaux INP, LaBRI, UMR5800, F-33400 Talence, France\and
\url{https://www.labri.fr/perso/dfigueir/}}{diego.figueira@cnrs.fr}{https://orcid.org/0000-0003-0114-2257}{}
\author{Pierre Lafourcade}{Univ. Bordeaux, CNRS, Bordeaux INP, LaBRI, UMR5800, F-33400 Talence, France\and
\url{https://www.labri.fr/perso/plafourca001/}}{pierre.lafourcade@u-bordeaux.fr}{https://orcid.org/0009-0004-4810-1289}{}
\authorrunning{Meghyn Bienvenu, Diego Figueira, and Pierre Lafourcade}
\Copyright{M. Bienvenu, D. Figueira, P. Lafourcade}



\funding{ANR grants INTENDED (ANR-19-CHIA-0014) and EXPAND (ANR-25-CE23-1215)
}

\nolinenumbers %
\pdfoutput=1 %
\ifcameraready
\else
    \hideLIPIcs  %
\fi

\theoremstyle{claimstyle}

\Crefname{claim}{Claim}{Claims}

\usepackage{dutchcal} %
\usepackage{xspace} %
\usepackage{booktabs} %
\usepackage{marginnote}
\usepackage[bibliography=common]{apxproof} %
\ifcameraready

\else
    
\fi

\usepackage{stmaryrd}

\usepackage[HTML]{xcolor}
\input{header/palette}
\input{header/kl-config}
\input{header/macros} %
\input{kl/kl-abbreviations}
\input{kl/kl-complexity}

\input{kl/kl-paper}

\AtEndPreamble{
\theoremstyle{theorem}

\newtheoremrep{proposition}{\sc Proposition}[section]
\newtheoremrep{lemma}{\sc Lemma}[section]
\newtheoremrep{theorem}{\sc Theorem}[section]
\newtheoremrep{claim}{\sc Claim}[section]
  \Crefname{claim}{claim}{claims}
  \Crefname{claim}{Claim}{Claims}
\newtheoremrep{corollary}{\sc Corollary}[section]
}

\EventEditors{Balder ten Cate and Maurice Funk}
\EventNoEds{2}
\EventLongTitle{29th International Conference on Database Theory (ICDT 2026)}
\EventShortTitle{ICDT 2026}
\EventAcronym{ICDT}
\EventYear{2026}
\EventDate{March 24--27, 2026}
\EventLocation{Tampere, Finland}
\EventLogo{}
\SeriesVolume{365}
\ArticleNo{20}

\begin{document}

\ifcameraready
    \relatedversiondetails{Full version \textnormal{\cite{thispaper-arxiv}}}{https://arxiv.org/abs/2601.04868} 
\fi

\maketitle

\begin{abstract}
    \input{abstract}

\end{abstract}

\bigskip

\input{knowledge-notice}

\ifcameraready\else
  \noindent\raisebox{-.4ex}{\HandRight}\ \ This article is based on the ICDT'26 paper \cite{thispaper}.
\fi

\bigskip

\input{sec-intro}

\input{sec-prelims}

\input{sec-signed}

\input{sec-global}

\input{sec-relwork-new}

\input{sec-conclusion}

    \bibliographystyle{plainurl}
\bibliography{long,bib}

\end{document}

%% file: header/palette.tex
\definecolor{Desire}{HTML}{eb3b5a} %
\definecolor{Boyzone}{HTML}{2d98da} %
\definecolor{Royal Blue}{HTML}{3867d6} %
\definecolor{NYC Taxi}{HTML}{f7b731} %
\definecolor{Beniukon Orange}{HTML}{fa8231}
\definecolor{Algal Fuel}{HTML}{20bf6b} %
\definecolor{Innuendo}{HTML}{a5b1c2} %
\definecolor{Twinkle Blue}{HTML}{d1d8e0} %
\definecolor{Blue Horizon}{HTML}{4b6584} %
\definecolor{Gloomy Purple}{HTML}{8854d0} %

\colorlet{cBlue}{Royal Blue}
\colorlet{cYellow}{NYC Taxi}
\colorlet{cOrange}{Beniukon Orange}
\colorlet{cGreen}{Algal Fuel}
\colorlet{cRed}{Desire}
\colorlet{cGrey}{Innuendo}
\colorlet{cDarkGrey}{Blue Horizon}
\colorlet{cLightGrey}{Twinkle Blue}
\colorlet{cPurple}{Gloomy Purple}

%% file: header/kl-config.tex
\usepackage[xcolor, cleveref, notion, quotation, electronic]{knowledge}
\usepackage{mathcommand}
\knowledgeconfigure{protect quotation={tikzcd}}
\knowledgeconfigure{diagnose line=true, diagnose bar=true}

\definecolor{Dark Ruby Red}{HTML}{580507}
\definecolor{Dark Blue Sapphire}{HTML}{053641}
\definecolor{Dark Gamboge}{HTML}{be7c00}

\IfKnowledgePaperModeTF{
}{
    \knowledgestyle{intro notion}{color={Dark Ruby Red}, emphasize}
    \knowledgestyle{notion}{color={Dark Blue Sapphire}}
    \hypersetup{
        colorlinks=true,
        breaklinks=true,
        linkcolor={cGreen}, %
        citecolor={cGreen}, %
        filecolor={cGreen}, %
        urlcolor={cGreen},
    }
    \IfKnowledgeElectronicModeTF{
    }{
        \knowledgeconfigure{anchor point color={Dark Ruby Red}, anchor point shape=corner}
        \knowledgestyle{intro unknown}{color={Dark Gamboge}, emphasize}
        \knowledgestyle{intro unknown cont}{color={Dark Gamboge}, emphasize}
        \knowledgestyle{kl unknown}{color={Dark Gamboge}}
        \knowledgestyle{kl unknown cont}{color={Dark Gamboge}}
    }
}

%% file: header/macros.tex
\usepackage{soul}%

\renewcommand{\epsilon}{\varepsilon}

\usepackage[backgroundcolor=orange!20, textcolor={Dark Ruby Red}, textsize=tiny
,disable %
]{todonotes}
\definecolor{green}{RGB}{0,120,0}
\definecolor{hlyellow}{RGB}{250, 250, 190}
\definecolor{diegoeditcolor}{RGB}{210,210,255}
\definecolor{pierreeditcolor}{RGB}{210,255,210}
\definecolor{meghyneditcolor}{RGB}{255, 225, 186}
\newcommand{\sidediego}[1]{}
\newcommand{\sidepierre}[1]{}
\newcommand{\sidemeghyn}[1]{}
\newcommand{\pierre}[1]{}
\newcommand{\meghyn}[1]{}

\newcommand{\diego}[1]{}
\definecolor{light-gray}{gray}{0.9}

\newcommand{\siderev}[1]{}
\newcommand{\rev}[1]{}

\newcommand{\proofcase}[1]{\noindent\colorbox{cLightGrey}{#1}~~}

\renewcommand{\leq}{\leqslant}
\renewcommand{\geq}{\geqslant}
\renewcommand{\emptyset}{\varnothing}

\knowledgenewrobustcmd{\dcup}{\mathop{\cmdkl{\uplus}}} %

\newcommand{\set}[1]{\{#1\}}

\newrobustcmd{\defeq}{\mathrel{\hat{=}}}
\newrobustcmd{\defcq}{\mathrel{{\vcentcolon}{\text{--}}}}
\newcommand{\eqdef}{\defeq}
\newcommand{\+}[1]{\mathcal{#1}}
\newrobustcmd{\Nat}{\mathbb{N}}
\newrobustcmd{\Reals}{\mathbb{R}}
\newrobustcmd\pset[1]{\+P(#1)} %
\newrobustcmd{\poly}{\mathop{\textrm{poly}}} %

\knowledgenewrobustcmd{\polyrx}{\mathrel{\cmdkl{\le_{\textit{poly}}}}} %
\knowledgenewrobustcmd{\polyeq}{\mathrel{\cmdkl{\equiv_{\textit{poly}}}}} %

\newrobustcmd\smashxrightarrow[1]{%
  \raisebox{-.04em}{$%
    \xrightarrow{\smash{\raisebox{-.1em}{%
      \tiny{#1}%
    }}}%
  $}%
}%

\knowledgenewrobustcmd{\homto}{\mathrel{\cmdkl{\smashxrightarrow{hom}}}}
\knowledgenewrobustcmd{\Chomto}[1][C]{\mathrel{\cmdkl{\xrightarrow{\smash{{#1}\textit{-hom}}}}}}

\knowledgenewrobustcmd{\Const}{\cmdkl{\textup{Const}}}
\knowledgenewrobustcmd{\Var}{\cmdkl{\textup{Var}}}
\knowledgenewrobustcmd{\ann}{\cmdkl{\nu}}
\knowledgenewrobustcmd{\oneann}{\cmdkl{\mathbf{1}}}
\knowledgenewrobustcmd{\adom}{\cmdkl{\textit{adom}}}

\knowledgenewrobustcmd{\atoms}{\cmdkl{\textit{atoms}}}
\knowledgenewrobustcmd{\vars}{\cmdkl{\textit{var}}}%
\knowledgenewrobustcmd{\const}{\cmdkl{\textit{const}}}
\knowledgenewrobustcmd{\mterms}{\cmdkl{\textit{term}}}
\knowledgenewrobustcmd{\arity}{\cmdkl{\textup{arity}}}
\knowledgenewrobustcmd{\class}{\mathcal{C}}
\knowledgenewrobustcmd{\evalgD}[2]{#1\cmdkl{\langle}#2\cmdkl{\rangle}} %
\knowledgenewrobustcmd{\evalaggD}[2]{#1\cmdkl{\langle}#2\cmdkl{\rangle}}
\knowledgenewrobustcmd{\evalqD}[2]{#1\cmdkl{\langle}#2\cmdkl{\rangle}}
\knowledgenewrobustcmd{\assto}{\mathrel{\cmdkl{\mapsto}}}

\knowledgenewrobustcmd{\CQ}{\cmdkl{\mathsf{CQ}}} %
\knowledgenewrobustcmd{\sjfCQ}{\cmdkl{\mathsf{sjf\text{-}CQ}}} %
\knowledgenewrobustcmd{\CQneq}{\cmdkl{\mathsf{CQ}^{\neq}}} %
\knowledgenewrobustcmd{\CQeq}{\cmdkl{\mathsf{CQ}^{=}}} %
\knowledgenewrobustcmd{\CQeqneq}{\cmdkl{\mathsf{CQ}^{\neq,=}}} %
\knowledgenewrobustcmd{\UCQ}{\cmdkl{\mathsf{UCQ}}} %
\knowledgenewrobustcmd{\CRPQ}{\cmdkl{\mathsf{CRPQ}}} %
\knowledgenewrobustcmd{\UCRPQ}{\cmdkl{\mathsf{UCRPQ}}} %
\knowledgenewrobustcmd{\ACQ}{\cmdkl{\mathsf{ACQ}}} %
\knowledgenewrobustcmd{\sjfACQ}{\cmdkl{\mathsf{sjf\text{-}ACQ}}} %

\knowledgenewrobustcmd{\core}{\cmdkl{\textit{core}}}
\knowledgenewrobustcmd{\emptytup}{\cmdkl{()}}
\knowledgenewrobustcmd{\hyperq}[1][q]{\cmdkl{\mathbf{G}_{#1}}}
\knowledgenewrobustcmd{\primalq}[1][q]{\cmdkl{\mathbf{G}^p_{#1}}}
\knowledgenewrobustcmd{\dimtup}{\cmdkl{\dim}}
\knowledgenewrobustcmd\vertex[1]{\cmdkl{V}(#1)}
\knowledgenewrobustcmd\edges[1]{\cmdkl{E}(#1)}

\knowledgenewrobustcmd{\maxSize}[1]{\cmdkl{\|}#1\cmdkl{\|}}
\knowledgenewrobustcmd{\sizeofD}[1][D]{\cmdkl{|}#1\cmdkl{|}}
\knowledgenewrobustcmd{\normOne}[1]{\cmdkl{\|}#1\cmdkl{\|_1}}
\knowledgenewrobustcmd{\normInf}[1]{\cmdkl{\|}#1\cmdkl{\|_{\infty}}}

\knowledgenewrobustcmd{\sizeofq}[1][q]{\cmdkl{|}#1\cmdkl{|}}
\knowledgenewrobustcmd{\sizeofgamma}[1][\gamma]{\cmdkl{|}#1\cmdkl{|}}

\knowledgenewrobustcmd{\reducesto}{\mathrel{\cmdkl{\leq_{\textit{poly}}}}}

\knowledgenewrobustcmd\bagmap{\cmdkl{\mathbf{b}}}
\knowledgenewrobustcmd\atommap{\cmdkl{\mathbf{a}}}
\knowledgenewrobustcmd\atommaplab{\cmdkl{\mathbf{\tilde a}}}
\knowledgenewrobustcmd\tagmap{\cmdkl{\mathbf{t}}}
\knowledgenewrobustcmd\tagmappath[1]{\cmdkl{\mathbf{t}[#1]}}
\newrobustcmd\tagmappathprime[1]{%
  \withkl{\kl[\tagmappath]}{%
    \cmdkl{\mathbf{t}'[#1]}%
  }%
}
\knowledgenewrobustcmd{\ghw}{\cmdkl{\textit{ghw}}}
\knowledgenewrobustcmd{\fghw}{\cmdkl{\textit{fghw}}}
\knowledgenewrobustcmd{\pghw}{\cmdkl{\textit{pghw}}}
\knowledgenewrobustcmd{\hw}{\cmdkl{\textit{hw}}}
\knowledgenewrobustcmd{\tw}{\cmdkl{\textit{tw}}}

\knowledgenewrobustcmd{\sjoin}[1][]{\mathop{\cmdkl{\ltimes_{#1}}}}
\knowledgenewrobustcmd{\costjoin}{\cmdkl{c_{sj}}}

\knowledgenewrobustcmd{\abotimes}[2]{\cmdkl{\bigotimes_{#1}}#2} %
\knowledgenewrobustcmd{\aboplus}[2]{\cmdkl{\bigoplus_{#1}}#2} %

\knowledgenewrobustcmd{\contr}{\cmdkl{\textit{contract}}}
\knowledgenewrobustcmd{\Mfacts}{\cmdkl{\mathbf{M}}}
\knowledgenewrobustcmd{\DBs}[1][\Sigma]{\cmdkl{\textup{DB}_{#1}}}
\knowledgenewrobustcmd{\evalPb}[1]{\cmdkl{\textup{\textsc{Eval-}}}#1}
\knowledgenewrobustcmd{\aug}[1]{{#1}^{\cmdkl{+}}}%
\knowledgenewrobustcmd{\restrictG}[2][G]{#1\cmdkl{[}{#2}\cmdkl{]}}
\knowledgenewrobustcmd{\evalCounting}[2]{\cmdkl{\#}{#1}\cmdkl{\langle}#2\cmdkl{\rangle}}
\knowledgenewrobustcmd{\counting}[1]{\cmdkl{\#}{#1}}
\knowledgenewrobustcmd{\PRel}{\cmdkl{P}}

\knowledgenewrobustcmd{\SNij}[1][i,j]{\cmdkl{\+S^N_{#1}}} %

\knowledgenewrobustcmd{\countGMS}[1]{\cmdkl{\#_{#1}^{\textsf{gms}}}}
\knowledgenewrobustcmd{\countMS}[1]{\cmdkl{\#_{#1}^{\textsf{ms}}}}
\knowledgenewrobustcmd{\countFMS}[1]{\cmdkl{\#_{#1}^{\textsf{fms}}}}
\knowledgenewrobustcmd{\evalCountMS}[1]{\cmdkl{\textsc{eval-}\#_{#1}^{\textsf{ms}}}}
\knowledgenewrobustcmd{\evalCountFMS}[1]{\cmdkl{\textsc{eval-}\#_{#1}^{\textsf{fms}}}}
\knowledgenewrobustcmd{\countAns}[1]{\cmdkl{\#_{#1}^{\textsf{hom}}}}
\knowledgenewrobustcmd{\evalCountAns}[1]{\cmdkl{\textsc{eval-}\#_{#1}^{\textsf{hom}}}}

\knowledgenewrobustcmd{\Ptime}{\cmdkl{\mathsf{P}}} %
\knowledgenewrobustcmd{\BPP}{\cmdkl{\mathsf{BPP}}} %
\knowledgenewrobustcmd{\FP}{\cmdkl{\mathsf{FP}}} %
\knowledgenewrobustcmd{\sP}{\cmdkl{\mathsf{\#P}}} %
\knowledgenewrobustcmd{\FPsP}{\cmdkl{\mathsf{FP}^\mathsf{\#P}}} %
\knowledgenewrobustcmd{\NP}{\cmdkl{\mathsf{NP}}} %
\knowledgenewrobustcmd{\sNP}{\cmdkl{\mathsf{\#NP}}} %
\knowledgenewrobustcmd{\FPsNP}{\cmdkl{\mathsf{FP}^{\mathsf{\#NP}}}} %
\knowledgenewrobustcmd{\PH}{\cmdkl{\mathsf{PH}}} %
\knowledgenewrobustcmd{\sPH}{\cmdkl{\mathsf{\#PH}}} %
\knowledgenewrobustcmd{\FPsPH}{\cmdkl{\mathsf{FP}^{\mathsf{\#PH}}}} %

\newcommand{\D}{\+{D}} %
\knowledgenewrobustcmd{\Dn}[1][\D]{#1_{\cmdkl{\subendo}}}
\knowledgenewrobustcmd{\Dx}[1][\D]{#1_{\cmdkl{\subexo}}}
\newcommand{\subendo}{\textup{\textsf{n}}}
\newcommand{\subexo}{\textup{\textsf{x}}}

\knowledgenewrobustcmd{\Sh}{\cmdkl{\mathrm{Sh}}} %
\knowledgenewrobustcmd{\Sym}{\cmdkl{\mathbb{P}}} %
\knowledgenewrobustcmd{\scorefun}[1][]{\cmdkl{\mathbf{\xi}_{#1}}} %
   \knowledgenewrobustcmd{\paramscorefun}[2][]{\cmdkl{\mathbf{\xi}}^{#2}_{#1}} %
   \knowledgenewrobustcmd{\PARAMscorefun}[2][]{\cmdkl{\Xi}^{#2}} %
\knowledgenewrobustcmd{\paramShapley}[2]{\cmdkl{\textup{SVC}^{#2}_{#1}}} %
   \knowledgenewrobustcmd{\dscorefun}[1][]{\cmdkl{\mathbf{\xi}}^{\cmdkl{\textsf{dr}}}_{#1}} %
   \knowledgenewrobustcmd{\Dscorefun}[1][]{\cmdkl{\Xi}^{\cmdkl{\textsf{dr}}}} %
   \newcommand{\minsupindex}{\textsf{ms}}
   \newcommand{\drasticindex}{\textup{\textsf{dr}}}
   \newcommand{\minmonsupindex}{\textup{\textsf{mps}}}
   \newcommand{\mondrasticindex}{\textup{\textsf{pdr}}}
   \knowledgenewrobustcmd{\drscorefun}[1][]{\cmdkl{\mathbf{\xi}}^{\cmdkl{\drasticindex}}_{#1}} %
   \knowledgenewrobustcmd{\msscorefun}[1][]{\cmdkl{\mathbf{\xi}}^{\cmdkl{\minsupindex}}_{#1}} %
   \knowledgenewrobustcmd{\monMSscorefun}[1][]{\cmdkl{\mathbf{\xi}}^{\cmdkl{\minmonsupindex}}_{#1}} %
   \knowledgenewrobustcmd{\monDRscorefun}[1][]{\cmdkl{\mathbf{\xi}}^{\cmdkl{\mondrasticindex}}_{#1}} %
   \knowledgenewrobustcmd{\MSscorefun}[1][]{\cmdkl{\Xi}^{\cmdkl{\minsupindex}}} %
   \knowledgenewrobustcmd{\msShapley}[1]{\cmdkl{\textup{SVC}^{\minsupindex}_{#1}}} %
\knowledgenewrobustcmd{\drShapley}[1]{\cmdkl{\textup{SVC}^{\drasticindex}_{#1}}} %
   \knowledgenewrobustcmd{\wscorefun}[2][]{\cmdkl{\mathbf{\xi}}^{#2}_{#1}} %
   \knowledgenewrobustcmd{\Wscorefun}[2][]{\cmdkl{\Xi}^{#2}} %
   \knowledgenewrobustcmd{\wShapley}[2]{\cmdkl{\textup{SVC}^{#2}_{#1}}} %
   \knowledgenewrobustcmd{\sscorefun}[1][]{\cmdkl{\mathbf{\xi}}^{\cmdkl{\textsf{s}}}_{#1}} %
   \knowledgenewrobustcmd{\Sscorefun}[1][]{\cmdkl{\Xi}^{\cmdkl{\textsf{s}}}} %
   \knowledgenewrobustcmd{\sShapley}[1]{\cmdkl{\textup{SVC}^{\textsf{s}}_{#1}}} %
   \knowledgenewrobustcmd{\sharpscorefun}[1][]{\cmdkl{\mathbf{\xi}}^{\cmdkl{\textsf{\#}}}_{#1}} %
   \knowledgenewrobustcmd{\SHARPscorefun}[1][]{\cmdkl{\Xi}^{\cmdkl{\textsf{\#}}}} %
   \knowledgenewrobustcmd{\sharpShapley}[1]{\cmdkl{\textup{SVC}^{\textsf{\#}}_{#1}}} %
   \knowledgenewrobustcmd{\pscorefun}[1][]{\cmdkl{\mathbf{\xi}}^{\cmdkl{\textsf{P}}}_{#1}} %
   \knowledgenewrobustcmd{\Pscorefun}[1][]{\cmdkl{\Xi}^{\cmdkl{\textsf{P}}}} %
   \knowledgenewrobustcmd{\pShapley}[1]{\cmdkl{\textup{SVC}^{\textsf{P}}_{#1}}} %
   \knowledgenewrobustcmd{\rscorefun}[1][]{\cmdkl{\mathbf{\xi}}^{\cmdkl{\textsf{R}}}_{#1}} %
   \knowledgenewrobustcmd{\Rscorefun}[1][]{\cmdkl{\Xi}^{\cmdkl{\textsf{R}}}} %
   \knowledgenewrobustcmd{\rShapley}[1]{\cmdkl{\textup{SVC}^{\textsf{R}}_{#1}}} %
   \knowledgenewrobustcmd{\mcscorefun}[1][]{\cmdkl{\mathbf{\xi}}^{\cmdkl{\textsf{MC}}}_{#1}} %
   \knowledgenewrobustcmd{\MCscorefun}[1][]{\cmdkl{\Xi}^{\cmdkl{\textsf{MC}}}} %
   \knowledgenewrobustcmd{\mcShapley}[1]{\cmdkl{\textup{SVC}^{\textsf{MC}}_{#1}}} %
   \knowledgenewrobustcmd{\sascorefun}[1][]{\cmdkl{\mathbf{\xi}}^{\cmdkl{\textsf{hom}}}_{#1}} %
   \knowledgenewrobustcmd{\SAscorefun}[1][]{\cmdkl{\Xi}^{\cmdkl{\textsf{hom}}}} %
   \knowledgenewrobustcmd{\saShapley}[1]{\cmdkl{\textup{SVC}^{\textsf{hom}}_{#1}}} %

\knowledgenewrobustcmd{\Minsups}[1]{\cmdkl{\mathsf{MS}_{#1}}} %

\knowledgenewrobustcmd{\shapscorefun}[1][]{\cmdkl{\mathbf{\xi}}^{\cmdkl{\textsf{SHAP}}}_{#1}} %

\knowledgenewrobustcmd{\Unif}[1][q]{\cmdkl{\mathbf{M}}_{#1}}

\knowledgenewrobustcmd{\qterms}{\cmdkl{\mathsf{terms}}}
\knowledgenewrobustcmd{\Equivs}{\cmdkl{\+E}}
\knowledgenewrobustcmd{\qE}[1][E]{\cmdkl{q_{#1}}}
\knowledgenewrobustcmd{\qEneq}[1][E]{\cmdkl{q_{#1}^{\neq}}}
\knowledgenewrobustcmd{\Punif}{\cmdkl{\mathbf{P}}_q}
\knowledgenewrobustcmd{\iso}{\cmdkl{\simeq}}
\knowledgenewrobustcmd{\Auto}[1]{\cmdkl{\textnormal{Aut}}(#1)}
\knowledgenewrobustcmd{\elimEq}[1]{\cmdkl{\widehat{#1}}}

\knowledgenewrobustcmd{\localfun}[1][q]{\cmdkl{\scorefun}^{\cmdkl{\textsf{\textup{LMS}}}}_{#1}}
\knowledgenewrobustcmd{\globalfun}[1][q]{\cmdkl{\scorefun}^{\cmdkl{\textsf{\textup{GMS}}}}_{#1}}

\knowledgenewrobustcmd{\LMSShapley}[1]{\cmdkl{\textup{SVC}^{\textsf{\textup{LMS}}}_{#1}}} %
\knowledgenewrobustcmd{\monMSShapley}[1]{\cmdkl{\textup{SVC}^{\minmonsupindex}_{#1}}} %
\knowledgenewrobustcmd{\monDRShapley}[1]{\cmdkl{\textup{SVC}^{\mondrasticindex}_{#1}}} %

\NewCommandCopy{\proofqedsymbol}{\qedsymbol}%
\newcommand{\exampleqedsymbol}{{$\vartriangleleft$}}%
\AtBeginEnvironment{proof}{\renewcommand{\qedsymbol}{\proofqedsymbol}}
\AtBeginEnvironment{nestedproof}{\renewcommand{\qedsymbol}{\nestedproofqedsymbol}}
\AtBeginEnvironment{example}{%
  \pushQED{\qed}\renewcommand{\qedsymbol}{\exampleqedsymbol}%
}
\AtEndEnvironment{example}{\popQED\endexample}

\knowledgenewrobustcmd{\myCheck}[1]{\cmdkl{\check #1}}
\knowledgenewrobustcmd{\myHat}[1]{\cmdkl{\hat #1}}

\knowledgenewrobustcmd{\Sigmapm}[1][\Sigma]{\cmdkl{{#1}^{\pm}}}
\knowledgenewrobustcmd{\spos}{\cmdkl{+}}
\knowledgenewrobustcmd{\sneg}{\cmdkl{-}}
\knowledgenewrobustcmd{\Dpm}[1][\D]{\cmdkl{#1^{\pm}}}
\knowledgenewrobustcmd{\Dpos}[1][\D]{\cmdkl{{#1}^+}}
\knowledgenewrobustcmd{\Dneg}[1][\D]{\cmdkl{{#1}^-}}
\knowledgenewrobustcmd{\querypm}[1][q]{\cmdkl{#1^{\pm}}}
\knowledgenewrobustcmd{\queryDpm}[2][q]{\cmdkl{#1_{#2}^{\pm}}}
\knowledgenewrobustcmd{\entailow}{\mathrel{\cmdkl{\Rrightarrow}}}
\knowledgenewrobustcmd{\queryyDpm}[1][\Dpm]{\cmdkl{q^\pm_{#1}}}
\knowledgenewrobustcmd{\modelspm}{\mathrel{\cmdkl{\models^{\pm}}}}
\knowledgenewrobustcmd{\paramShapleypm}[2]{\cmdkl{\pm\paramShapley{#1}{#2}}}
\knowledgenewrobustcmd{\ShapleypmMinSup}[1]{\cmdkl{\pm\textup{SVC}}^{\cmdkl{\minsupindex}}_{#1}}
\knowledgenewrobustcmd{\ShapleypmDrastic}[1]{\cmdkl{\pm\textup{SVC}}^{\cmdkl{\drasticindex}}_{#1}}

\knowledgenewrobustcmd{\entailcw}[1][C]{\mathrel{\cmdkl{\Rrightarrow_{#1}}}}
\knowledgenewrobustcmd{\modelspmcw}[1][C]{\mathrel{\cmdkl{\models_{#1}^{\pm}}}}

\knowledgenewrobustcmd{\Gaifman}[1]{\cmdkl{\mathcal{G}}(#1)}

\knowledgenewrobustcmd{\psubmu}[2]{\cmdkl{#1_{#2}}} 
\knowledgenewrobustcmd{\Atsubp}[1]{\cmdkl{\mathit{At}_{#1}}}
\knowledgenewrobustcmd{\Atsubpplus}[1]{\cmdkl{\mathit{At}^+_{#1}}}

\knowledgenewrobustcmd{\Dpmq}[1][q]{\cmdkl{\D^\pm_{#1}}}
\knowledgenewrobustcmd{\NegRelsq}[1][q]{\cmdkl{\+N_{#1}}}

%% file: kl/kl-abbreviations.tex
\knowledge{text={resp.}, italic}
  | resp

\knowledge{text={i.e.}, italic}
  | ie

\knowledge{text={I.e.}, italic}
  | Ie

\knowledge{text={s.t.}, italic}
  | st

\knowledge{text={e.g.}, italic}
  | eg

\knowledge{text={vs.}, italic}
  | vs

\knowledge{text={w.r.t.}, italic}
  | wrt

\knowledge{text={a.k.a.}, italic}
  | aka

\knowledge{text={w.l.o.g.}, italic}
  | wlog

\knowledge{text={cf.}, italic}
  | cf

\knowledge{text={iff}, italic}
  | iff

\knowledge{text={r.e.}, italic}
  | r.e.
  | re

%% file: kl/kl-complexity.tex
\knowledge{wrap=\textsf}
  | NL

\knowledge{notion, text={paraNL}, wrap=\textsf}
  | para-NL
  | paraNL

\knowledge{notion, text={\textup{FPT}}, wrap=\textsf}
  | FPT

\knowledge{notion}
  | fixed-parameter tractable

\knowledge{text={ExpSpace}, wrap=\textsf}
  | EXPSPACE
  | ExpSpace

\knowledge{text={\ensuremath{\textsf{\textup{AC}}_0}}, wrap=\textsf}
  | AC0

  \knowledge{text={2ExpSpace}, wrap=\textsf}
  | 2EXPSPACE
  | 2ExpSpace

\knowledge{text={PSpace}, wrap=\textsf}
  | PSpace
  | PSPACE

  \knowledge{text={PTime}, wrap=\textsf}
  | PTime
  | ptime

\knowledge{text={FP}, wrap=\textsf}
  | FP

\knowledge{text={NP}, wrap=\textsf}
  | NP

  \knowledge{text={P}, wrap=\textsf}
  | P

  \knowledge{notion, text={\ensuremath{\textup{W}[1]}}, wrap=\textsf}
  | W1

  \knowledge{text={\ensuremath{\# \textsf{P}}}, wrap=\textsf}
  | shP

  \knowledge{text={\ensuremath{\# {\cdot}\textsf{NP}}}, wrap=\textsf}
  | shNP

  \knowledge{text={\ensuremath{\textsf{P}^{\# \textsf{P}}}}, wrap=\textsf}
  | PshP
  
\knowledge{text={\ensuremath{\textsf{FP}^{\# \textsf{P}}}}, wrap=\textsf}
  | FPshP

\knowledge{text={\ensuremath{\Pi^p_2}}, wrap=\textsf}
  | PiP2
  | Pi2
  | pi2
  | Pitwo
  | PiTwo
  | pitwo

\knowledge{url={https://en.wikipedia.org/wiki/Savitch\%27s_theorem}}
  | Savitch's Theorem

%% file: kl/kl-paper.tex
\knowledge{notion}
 | notion@notice
 | definition@notice

\knowledge{notion}
| $1$-Turing reductions

\knowledge{notion}
| acyclic
| acyclic CQ
| acyclic CQs
| Acyclic CQs
| ACQ
| ACQs

\knowledge{notion}
| homomorphic image

\knowledge{notion}
| answer
| answers

\knowledge{notion}
| atom
| atoms
| (relational) atom

\knowledge{notion}
| automorphisms
| automorphism

\knowledge{notion}
| Banzhaf power index

\knowledge{notion}
| b-domination witness
| b-domination witnesses

\knowledge{notion}
| Boolean
| Boolean query
| Boolean queries

\knowledge{notion}
| Boolean CQ

\knowledge{notion}
 | bounded
 | unbounded
 | boundedness

\knowledge{notion}
 | bounded negative arity

\knowledge{notion}
| canonical database

\knowledge{notion}
 | conjunctive queries
 | CQs
 | CQ
 | conjunctive query

\knowledge{notion}
| constants
| constant

\knowledge{notion}
| core
 | cores

\knowledge{notion, text={\ensuremath{\text{CQ}^{=}}}}
| CQeq

\knowledge{notion, text={\ensuremath{\text{CQ}^{\neq,=}}}}
| CQeqneq

\knowledge{notion, text={\ensuremath{\text{CQ}^{\neq}}}}
| CQneq

\knowledge{notion, text={\ensuremath{\text{UCQ}^{\neq}}}}
| UCQneq

\knowledge{notion, text={\ensuremath{\text{CQ}^{\neq}_{\set{\land,\lnot}}}}}
| negCQ

\knowledge{notion, text={\ensuremath{\text{sjf-CQ}^{\lnot}}}}
| sjfCQneg
| sjf-CQneg

\knowledge{notion, text={\ensuremath{\text{CQ}^{\lnot}}}}
| CQneg

\knowledge{notion, text={\ensuremath{\text{CQ}^{\textsf{g}\lnot}}}}
| gCQneg

\knowledge{notion, text={\ensuremath{\text{UCQ}^{\lnot}}}}
| UCQneg

\knowledge{notion, text={\ensuremath{\text{sUCQ}^{\lnot}}}}
| sUCQneg
\knowledge{notion, text={\ensuremath{\text{gUCQ}^{\lnot}}}}
| gUCQneg

\knowledge{notion}
| Gaifman graph

\knowledge{notion}
| database
| databases

\knowledge{notion}
| active domain

\knowledge{notion}
| data complexity

\knowledge{notion}
| inequality atoms
| inequality atom

\knowledge{notion}
| domination witness
| domination witnesses

\knowledge{notion}
| drastic-Shapley value
| drastic-Shapley

\knowledge{notion}
| Drastic-Shapley measure for signed facts

\knowledge{notion}
| endogenous

\knowledge{notion}
| equality atoms
| equality atom

\knowledge{notion}
| equivalence relations

\knowledge{notion}
| exogenous
| exogenous facts

\knowledge{notion}
| fact
| facts

\knowledge{notion}
| features

\knowledge{notion}
| free

\knowledge{notion}
| full
| full CQ

\knowledge{notion}
| fully polynomial randomized approximation scheme
| FPRAS

\knowledge{notion}
| generalized hypertree width

\knowledge{notion}
| homomorphisms
| homomorphism
| homomorphically

\knowledge{notion}
| homomorphism-closed

\knowledge{notion}
| homomorphisms@eq
| homomorphism@eq
| homomorphism-closed@eq
| homomorphically@eq

\knowledge{notion}
| independent set

\knowledge{notion}
| induces@support
| induced@support
 | inducing@support
 | inducing@support

\knowledge{notion}
| inverse weight function
| invw

\knowledge{notion}
| irrelevant
| relevant
| irrelevance
| relevance

\knowledge{notion}
| isomorphism
 | isomorphic

\knowledge{notion}
| join tree

\knowledge{notion}
| \mcscorefun
| \MCscorefun
| \mcShapley

\knowledge{notion}
| minimal supports
| minimal support

\knowledge{notion}
| null@player
 | null player
 | null players

\knowledge{notion}
| numeric
| numeric queries

\knowledge{notion}
| path query
 | path queries

\knowledge{notion}
| perfect matching
| perfect matchings

\knowledge{notion}
| perfect matching counting

\knowledge{notion}
| polynomial hierarchy

\knowledge{notion}
| \pscorefun
| \Pscorefun
| \pShapley

\knowledge{notion}
| queries
| query
 | (non-numeric) query
 | non-numeric queries
| non-numeric

\knowledge{notion}
| regular path query

\knowledge{notion}
| relation name
| RPQ

\knowledge{notion}
| responsibility measure
| responsibility measures
| quantitative measure
| measures

\knowledge{notion}
| reversible
| reversible WSMS

\knowledge{notion}
| \rscorefun
| \Rscorefun
| \rShapley

\knowledge{notion}
| \sascorefun
| \SAscorefun
| \saShapley

\knowledge{notion}
| satisfies

\knowledge{notion}
| satisfying assignments
| satisfying assignment

\knowledge{notion}
| schema
| schemas
| (relational) schema

\knowledge{notion}
| Shapley-like

\knowledge{notion}
| SHAP score

\knowledge{notion}
| self-join free
| self-joins
| sjf-CQ

\knowledge{notion}
| self-join width

\knowledge{notion}
| Shapley value
| Shapley values

\knowledge{notion}
| \sharpscorefun
| \SHARPscorefun
| \sharpShapley

\knowledge{notion}
| \sscorefun
| \Sscorefun
| \sShapley

\knowledge{notion}
| support
| supports

\knowledge{notion}
| terms
| term

\knowledge{notion}
| tractable weight function
| tractable weight functions
| tractable@wf

\knowledge{notion}
| Turing reductions
| polynomial-time Turing-reductions

\knowledge{notion}
| unions of conjunctive queries
| UCQ
| UCQs

\knowledge{notion}
| variables
| variable

\knowledge{notion}
| VertexCover
| vertex cover

\knowledge{notion}
| wealth function
| wealth functions

\knowledge{notion}
| weight function
| weight functions

\knowledge{notion}
| \wscorefun
| \Wscorefun
| \wShapley

\knowledge{notion}
| WSAXp

\knowledge{notion}
| WSMS
| WSMSs
| weighted sums of minimal supports
| weighted sum of minimal supports

\knowledge{notion}
| repair
| repairs

\knowledge{notion}
| subquery

\knowledge{notion}
 | semantics

\knowledge{notion}
 | signed facts
 | signed fact

\knowledge{notion}
 | positive support
 | minimal positive support
 | positive supports

\knowledge{notion}
 | $S$-monotone support
 | $S$-monotone supports
 | $\D$-monotone support
 | $\D$-monotone supports

\knowledge{notion}
 | positive-drastic-Shapley

\knowledge{notion}
 | MPS-Shapley

\knowledge{notion}
 | monotone query
 | monotone queries
 | monotone
 | monotonicity

\knowledge{notion}
 | minimal signed support
 | minimal signed supports

\knowledge{notion}
 | signed support
 | signed supports

\knowledge{notion}
 | counter-support
 | counter-supports

 \knowledge{notion}
 | safe-bounded

 \knowledge{notion}
 | positive atom
 | positive atoms
\knowledge{notion}
 | positive fact
 | positive facts

 \knowledge{notion}
 | negated atom
 | negated atoms
\knowledge{notion}
 | negative atom
 | negative atoms
\knowledge{notion}
 | negative fact
 | negative facts

 \knowledge{notion}
 | impact-relevant
 | impact-relevance
\knowledge{notion}
 | positive impact
\knowledge{notion}
 | positively impactful

\knowledge{notion}
 | negative impact
 | negative impacts
 
\knowledge{notion}
 | negatively impactful

\knowledge{notion}
 | impact
 | impacts

 \knowledge{notion}
 | safe negations
 | safe negation

\knowledge{notion}
 | non-hierarchical neg-path

\knowledge{notion}
 | Global drastic-Shapley

\knowledge{notion}
 | drastic-Shapley

\knowledge{notion}
 | Local drastic-Shapley
 | local drastic-Shapley

\knowledge{notion}
 | Global WSMS

\knowledge{notion}
 | Local WSMS
\knowledge{notion}
  | minimal
 | minimality

\knowledge{notion}
 | relational signature
 | signature

\knowledge{notion}
 | arity

\knowledge{notion}
 | guarded negation

\knowledge{notion}
 | `local' minimal monotone support
 | local minimal monotone support
 | local minimal monotone supports

\knowledge{notion}
 | signed-relevance
 | signed-relevant

\knowledge{notion}
 | positive-relevance
 | positive-relevant
 | non-positive-relevant
\knowledge{notion}
 | domain-independent
 | domain-dependent
 | domain-independence

\knowledge{notion}
 | monotone-bounded
 | monotone-boundedness 
 | monotone-unbounded

\knowledge{notion}
| mergeable atom
| mergeable atoms
| mergeable
| non-mergeable

\knowledge{notion}
| self-join-width
\knowledge{notion}
| combined complexity

\knowledge{notion}
| safe UCQs
| safe UCQ

\knowledge{notion}
| safe@UCQneg
\knowledge{notion}
 | entailment semantics
\knowledge{notion}
 | Local MS-Shapley

\knowledge{notion}
 | MS-Shapley

\knowledge{notion}
 | Global MS-Shapley

\knowledge{notion}
 | MS-Shapley measure for signed facts

\knowledge{notion}
 | Drastic Shapely measure for signed facts

%% file: abstract.tex
We contribute to the recent line of work on responsibility measures that quantify the contributions of database facts to obtaining a query result. 
In contrast to existing work which has almost exclusively focused on monotone queries, 
here we explore how to define responsibility measures for unions of conjunctive queries with negated atoms (UCQ$^\neg$s). 
After first investigating the question of what constitutes a reasonable notion of qualitative explanation or relevance for queries with negated atoms, 
we propose two approaches, one assigning scores to (positive) database facts and the other also considering negated facts. 
Our approaches, which are orthogonal to the previously studied score of  Reshef et al. \cite{ReshefKL20},
 can be used to lift previously studied scores for monotone queries, known as drastic Shapley and weighted sums of minimal supports (WSMS), 
to UCQ$^\neg$s. %
We investigate the data and combined complexity of the resulting measures, notably showing that %
the WSMS measures are tractable %
in data complexity for all UCQ$^\neg$s and %
further establishing tractability in combined complexity for suitable classes of conjunctive queries with negation.

%% file: knowledge-notice.tex
\noindent
\raisebox{-.4ex}{\HandRight}\ \ This pdf contains internal links: clicking on a "notion@@notice" leads to its \AP ""definition@@notice"".

%% file: sec-intro.tex
\section{Introduction}
\label{sec:intro}
Responsibility measures assign scores to database facts based upon how much they contribute to the obtention of a given query answer, thereby providing a quantitative notion of explanation for query results. There has been significant recent interest in defining and computing responsibility measures \cite{DBLP:journals/pvldb/MeliouGMS11,salimiQuantifyingCausalEffects2016,livshitsShapleyValueTuples2021,ReshefKL20,deutchComputingShapleyValue2022a,KhalilK23ShapleyRPQ,KaraOlteanuSuciuShapleyBack,abramovichBanzhafValuesFacts2024,ourpods24,karmakarExpectedShapleyLikeScores2024,ourpods25}, largely focusing on classes of monotone queries such as (unions of) conjunctive queries and regular path queries. 
In the present paper, we investigate responsibility measures for %
queries \emph{with negations}. 
The most basic such class is that of conjunctive queries allowing for negated atoms or inequalities -- here denoted by "CQneg" -- such as $q(y) = \exists x ~ R(x,y) \land x \neq y \land \lnot R(y,x)$. Our results and definitions will focus on the class "CQneg" and its natural extension "UCQneg" with unions.

We begin our study by first asking a fundamental question (of independent interest): What are the database entities over which the responsibility should be distributed? In other words: What are the so-called `relevant' facts for a given query being true?
Usually, `relevant' facts are defined as those being part of an `explanation'. But then: what is a good notion of an `explanation' for queries with negation?
In the case of (Boolean) monotone queries -- such as conjunctive queries without negated atoms -- a simple yet natural %
notion of qualitative explanation is that of a \emph{minimal support},
"ie", a subset-minimal set of facts making the query true. Through this lens, a fact is relevant to a monotone query whenever it is contained in a minimal support thereof, which witnesses its active participation in making the query true\footnote{This notion of relevance for monotone queries admits multiple equivalent characterizations and notably coincides with the notion of \emph{actual cause} from causal responsibility \cite{DBLP:journals/pvldb/MeliouGMS11}.}. 
However, we shall argue that for non-monotone queries, there is no unique sensible definition neither of `relevance' nor of `explanation', even when restricted to the rather basic class of "CQneg". Instead, we shall discuss several candidate definitions and compare their properties, leading us to focus on two notions of support (each defining a corresponding notion of relevance): "signed supports" and "positive supports". 

Let us begin with a concrete example, depicted in \Cref{fig:ex-recipe}. %
The database $\D$ we consider has a single binary relation $I$ for \emph{`has as Ingredient'}, and for now we focus on the Boolean query $q$ expressing ``there exists a recipe with \textit{fish} but not \textit{meat} as ingredient''.
Note that $\D \models q$ via the (only) "satisfying assignment"
$\set{x \mapsto mm}$.
One might naturally consider %
$\set{I(mm,\textit{fish})}$ as an explanation for $\D \models q$, since it is a minimal sub-database that makes $q$ true.
But then, by the same token, one might be led to conclude that $\set{I(mp,\textit{fish})}$ is an explanation, even if it does not participate in any "satisfying assignment"!
While one may conceive of scenarios in which it may be reasonable to 
regard $\set{I(mp,\textit{fish})}$ as a valid explanation -- and thus to consider $I(mp,\textit{fish})$ as being `relevant' -- %
we would argue that it is desirable 
to define notions of explanation and relevance under which a %
 fact is deemed `non-relevant' (and hence be assigned zero responsibility) whenever it does not participate in the satisfaction of the query, as is the case of  $I(mp,\textit{fish})$.

\begin{figure}[tb]
\centering
\includegraphics[width=\textwidth]{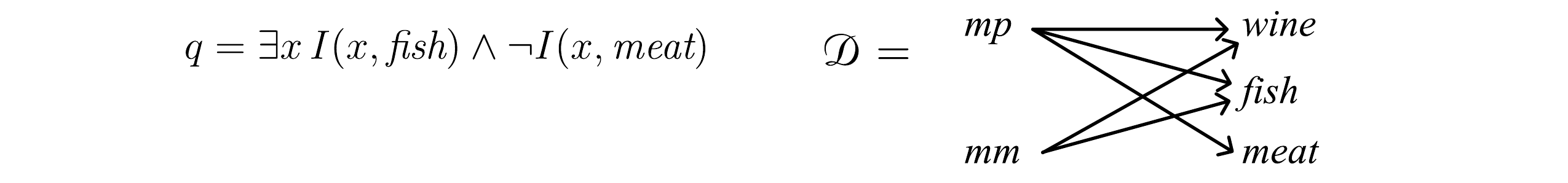}
\caption{Example database and queries, based on recipes from \cite{escoffierGuideCulinaireAidememoire1903}. All edges are $I$-tuples, and 
recipe names \emph{Matelote Meunière} and \emph{Matelote Pochouse} are abbreviated to $mm$ and %
$mp$ respectively.}
\label{fig:ex-recipe}
\end{figure}

The reason why we may consider $\set{I(mm,\textit{fish})}$ to be an explanation but not $\set{I(mp,\textit{fish})}$ lies, of course, in the \emph{absence} of $I(mm,\textit{meat})$ in the database. 
Hence, one way of explaining why the query is true is by using the information of both the presence and the absence of database facts. 
From this perspective, an explanation for the example query would be $\set{+I(mm,\textit{fish}),-I(mm,\textit{meat})}$, which can be read ``there is a fact $I(mm,\textit{fish})$ in the database, but no fact $I(mm,\textit{meat})$''. 
The `"signed facts"' approach thus postulates that, since the satisfaction of the query depends both on the presence and on the absence of facts, the responsibility should be distributed over "signed facts" (rather than just `positive' facts).
This idea can be formalized by viewing a database as a set of signed facts (using signed relations $+ P$ and $- P$, for each original relation $P$)
and defining minimal "signed supports" as minimal sets of signed facts that satisfy the query (also rephrased using the signed relations), where a fact is "signed-relevant" if it belongs to some such set.
Conveniently, this approach yields a direct reduction to the monotone case, thereby enabling the reuse %
of existing results and algorithms for %
monotone queries.

It may not always be desirable, however, to attribute responsibility to entities with no concrete materialization, such as absent facts. 
Further, it may not always be practical or realistic to use this "signed facts" approach, due to the sheer number of negated facts that need to be considered, %
which leads us to our second proposal. Returning to our example, 
we may consider $S=\set{I(mm,\textit{fish})}$  to be an explanation 
because there is a variable assignment that maps the positive query atoms to $S$ in such a way that
the negated atoms are satisfied \emph{in the full database} ("ie", the variable assignment continues to witness  the satisfaction of the query once the remaining database facts are incorporated). %
We call subsets satisfying this condition (minimal) "positive supports" since only the original (positive) database facts may take part in the support. 
Observe that $\set{I(mp,\textit{fish})}$ does not count as a "positive support", since the only way to satisfy $q$ in $\set{I(mp,\textit{fish})}$ is to map $x$ to $mp$, but this assignment does not satisfy $\neg I(x,\textit{meat})$  in 
 $\D=\set{I(mp,\textit{fish}),I(mp,\textit{meat})}$. 
When taking "minimal" "positive supports" as our notion of explanation, we thus obtain a unique explanation for $q$ in $\D$, namely $\set{I(mm,\textit{fish})}$. We will then say that $I(mm,\textit{fish})$ is the sole "positive-relevant" fact for $q$ in $\D$.

These are the two viewpoints advanced in this paper: the "signed facts" approach and "positive support" approach. For both approaches, we show how the notions of support and relevance can be used to adapt existing Shapley-value-based responsibility measures to unions of "CQneg"s. %
We consider in particular the `drastic-Shapley' responsibility measure, a well-studied measure \cite{livshitsShapleyValueTuples2021} for Boolean non-numeric queries which in the case of a monotone query $q$ assigns a fact $\alpha \in \D$ %
 the proportion of the linear orderings $(\D,<)$ for which $\set{\beta \in \D : \beta < \alpha } \not\models q$ and $\set{\beta \in \D : \beta \leq \alpha } \models q$.
We also consider the `"MS-Shapley"' measure,  recently introduced in \cite{ourpods25} as an alternative Shapley-based measure for Boolean monotone queries and shown to enjoy appealing theoretical and algorithmic properties. It assigns a fact $\alpha \in D$  the sum, over all "minimal supports" $S\subseteq \D$ of $q$ containing $\alpha$, of $\frac{1}{|S|}$ and it represents %
the simplest instance of the family of  "weighted sums of minimal supports" (or "WSMS") measures  \cite{ourpods25}. 

\subparagraph*{Related Work}
The most relevant prior work is undoubtedly that of Reshef, Kimelfeld, and Livshits \cite{ReshefKL20}, which studies a responsibility measure that adapts the `drastic-Shapley' for monotone queries to handle conjunctive queries with negation (but without inequalities). 
The notion of relevance underlying their measure %
is in terms of the `impact' that the introduction of the fact makes in a sub-database and is orthogonal to the notions of "signed-relevance" and "positive-relevance" that we propose in this paper. 
In particular, a fact such as $I(mp,\textit{fish})$ (from the example in \Cref{fig:ex-recipe}) would be deemed relevant and receive a non-null responsibility score in their approach. We shall dedicate \Cref{sec:local-monotone} to a comparison with their approach. %

\subparagraph*{Contributions} The foremost contribution of our work is of a conceptual nature, namely, the introduction of novel responsibility measures for (unions of) conjunctive queries with negations, which are grounded in simple and intuitive notions of supports and relevance. Our main technical contribution is the study of the computational complexity of computing the new measures, focusing primarily on the identification of tractable cases.  %

In more detail, 
\Cref{sec:signed}  formalizes the notions of "signed support" and "signed-relevance", where "negative facts" are treated as first-class citizens. We show that for "CQneg"s these notions can be equivalently defined in terms of logical entailment, but that for "UCQneg"s, the entailment-based definition yields different notions that we argue are less intuitive. We therefore choose "signed support"s as the basis for our definitions of novel "drastic-Shapley" and "MS-Shapley"-like measures for "UCQneg"s. In Sections \ref{ss:ms-signed} and \ref{ss:dr-signed}, we present several complexity results that identify classes of queries for which we can tractably compute the latter measures, which leverage the existence of a straightforward reduction to existing measures for queries without negation. In particular, we show that the `signed' version of the "MS-Shapley" measure (and more generally, "WSMS" measures) enjoys tractable ("FP") data complexity for all "UCQneg"s.

\Cref{sec:global-monotone} formalizes the "positive support" approach. Here again we consider and reject an alternative formalization -- based upon monotonicity of supports -- which we show to be less intuitive. We thus adopt "positive supports" (and "positive-relevance") 
as the basis for defining  new "drastic-Shapley" and "MS-Shapley"-like measures for "UCQneg"s, which attribute responsibility only to the (positive) database facts. In \Cref{sec:global-wsms} and \Cref{sec:global-drshapley}, we present complexity results for both measures, including a general tractability result in data complexity for the `positive' variant of the  "MS-Shapley" measure.

Finally, in \Cref{sec:local-monotone},  we discuss the connection with the closest prior work \cite{ReshefKL20} 
and compare their definitions with our own, both conceptually and algorithmically. We prove in particular that the notion of relevance that underlies the measure of \cite{ReshefKL20} is  incomparable with both "positive-relevance" and "signed-relevance".

%% file: sec-prelims.tex
\section{Preliminaries}
\label{sec:prelims}

\AP
We fix disjoint infinite sets $\intro*\Const$, $\intro*\Var$ of ""constants"" and ""variables"", respectively. 
For any syntactic object $O$ ("eg" database, query), we will use \AP$\intro*\vars(O)$ and $\intro*\const(O)$ 
to denote the sets of "variables" and "constants" contained in $O$.

\AP A ""(relational) schema"" %
is a finite set of relation symbols, each symbol $R$ associated with a (positive) ""arity"" $\intro*\arity(R)$.
\AP
A ""(relational) atom"" over a "schema" $\Sigma$ takes the form $R(\bar t)$ where $R$ is a ""relation name"" from $\Sigma$ of some "arity" $k$, and $\bar t \in (\Const \cup \Var)^k$.
\AP
A ""fact"" is an "atom" which contains only "constants".
\AP
A ""database"" $\D$ over a "schema" $\Sigma$ is a finite set of "facts" over $\Sigma$ we call $\const(\D)$ its ""active domain"".
When we speak of the size of a database or set of atoms, we mean the number of atoms (or facts) it contains. 

\subparagraph*{Conjunctive Queries with Negations}
\AP
A conjunctive query with negated atoms and inequalities, or ""CQneg"", is a first-order query of the form
$\exists \bar x ~ \alpha_1 \land \dotsb \land \alpha_n$
where each $\alpha_i$ can be (1) a (positive) "atom", (2) a ""negated atom"" of the form $\lnot \alpha$ where $\alpha$ is an "atom", or (3) an ""inequality atom"" of the form $t \neq t'$, where $t,t'$ are "terms". We further impose the condition that %
each variable appearing in a "negated@negated atom" or "inequality atom"  must also appear in a positive "atom", commonly known as the restriction to ""safe negations"". 
We shall also consider ""UCQneg""s, which are defined as finite unions of "CQneg"s, as well as  ""CQneq""s and ""UCQneq""s obtained by disallowing "negated atoms" in  "CQneg"s and "UCQneg"s,  and ""CQ"" and ""UCQ"" obtained by further restricting "CQneq"s and "UCQneq"s to  have only "positive atoms". 

\AP
A query is ""Boolean"" if it has no free variables, in which case we use the notation  $\D \models q$ to indicate that $q$ holds (or is satisfied) in $\D$. 
A ""satisfying assignment"" for a "UCQneg" $q$ on a "database" $\D$ is as expected a valuation $\nu : \vars(q) \to \D$ making the query true on $\D$. When $q$ is a "CQ", we shall also call the "satisfying assignment" a ""homomorphism"" (and write $q \intro*\homto \D$) and we extend the definition to sets of "atoms" in the obvious way.
\AP
A "Boolean" query is ""monotone"" if $\D \models q$ and $\D \subseteq \D'$ implies $\D' \models q$ for all "databases" $\D,\D'$. A ""support"" of $q$ in $\D$ is any set $S \subseteq \D$ such that $S \models q$, and such a subset $S$ is a ""minimal support"" if it is further subset-"minimal"\footnote{Throughout the paper, we will always use \AP""minimal"" to mean minimal "wrt" set inclusion.}  among all "supports" of $q$ in $\D$.

\AP
We recall some relevant structural restricted classes of "CQs". The ""generalized hypertree width"" of a "CQ" is a classic measure of tree-likeness. We refer readers to \cite[Definition~3.1]{GottlobGLS16} for a definition but simply note that bounded "generalized hypertree width" is a sufficient condition for tractable "CQ" evaluation and notably generalizes the class of `acyclic' "CQs", which correspond to "CQs" of "generalized hypertree width" 1.
We say that two distinct atoms $\alpha,\beta$ of a "CQ" $q$ are \AP""mergeable""  if there are two "homomorphisms" $h_\alpha : \alpha \homto \D$, $h_\beta : \beta \homto \D$
to an arbitrary "database" $\D$  such that $h_\alpha(\alpha)=h_\beta(\beta)$ (in particular they must have the same "relation name"). 
An "atom" $\alpha$ is (individually) "mergeable" if it is "mergeable" with some other "atom" in $q$.
The \AP""self-join width"" of a "CQ" $q$, defined in \cite{ourpods25}, is the cardinality of $\set{ t \in \qterms(\alpha) : \alpha \text{ is a "mergeable" "atom" of $q$}}$. The class of "CQs" with  "self-join width" 0 generalizes the class of 
""self-join free"" "CQs" (\reintro{sjf-CQ}),  defined as those "CQs" 
which do not contain two distinct atoms with the same relation name. To see why the two classes do not coincide, consider the query $\exists x y \, R(a,x) \wedge R(b,y)$ which has no "mergeable" "atoms" (hence "self-join width" 0) but does have a self-join.

\subparagraph*{Responsibility Measures via Shapley Values}
\AP
We shall be interested in ""responsibility measures"", defined as functions $\phi$ which take as input a "database" $\D$, a (possibly non-Boolean) "query" $q$, an answer $\bar a$ to $q$ in $\D$, and a fact $\alpha \in \D$, and which output a quantitative score measuring how much $\alpha$ contributes to $\bar{a} \in q(\D)$. 
We may however simplify the presentation by replacing the input answer $\bar{a}$ and (possibly non-Boolean) query $q$ by the associated Boolean query $q(\bar{a})$, defined by letting $\D \models q(\bar{a})$ iff $\bar{a} \in q(\D)$. 
In this manner, we can eliminate the answer tuple from the arguments of the responsibility measure and work instead with ternary "responsibility measures", "ie", $\phi(\D,q,\alpha)$ with "Boolean" $q$. 
Henceforth, we shall thus \emph{assume "wlog" that the input query 
is always "Boolean"}.

\AP
We study responsibility measures based upon the well-known "Shapley value", which was originally defined for cooperative games but can be translated into the database setting by modelling the query as a "wealth function". %
In the context of databases, a ""wealth function"" for a ("Boolean") query $q$ and "database" $\D$ is a function $\xi : 2^{\D} \to \Reals$ outputting a number for each set $S$ of "facts" in the database,
which captures in some manner how $S$ contributes to the satisfaction of $q$. 
The "Shapley value" can then be used for transforming such a function on subsets of $\D$ %
into a %
responsibility measure for individual facts in $\D$.
Concretely, the "Shapley value" of a "fact" $\alpha \in \D$ and a "wealth function" $\xi$ (encoding the query), denoted by $\Sh(\D, \xi, \alpha)$, is the average wealth contribution of $\alpha$ over all the linear orderings $\+O$ of $\D$, given by
\[
    \intro*\Sh(\D, \xi, \alpha) \eqdef \frac{\sum_{(\+D,<) \in \+O} \xi(\set{\beta \in \D : \beta < \alpha } \cup \set\alpha) - \xi(\set{\beta \in \D : \beta < \alpha })}{|\+O|}.
\]
Most work to date on Shapley-based responsibility measures employs  %
 the "drastic@drastic-Shapley" "wealth function" $\intro*\drscorefun[q]$ defined by setting %
 $\drscorefun[q](S)=1$ 
 if $S$ is a "support" of $q$ ("ie", $S \models q$) or $0$ 
 otherwise.%
 \footnote{In the original formulation of \cite{livshitsShapleyValueTuples2021}, the "wealth function" outputs $1$ "iff" $S \cup X \models q$ and $X \not \models q$, where $X \subseteq \D$ is a special set of facts called `exogenous'. For simplicity, our work considers plain databases (without exogenous facts), 
see \Cref{rk:exogenous}, so we have adjusted the function accordingly. }
More recently, \cite{ourpods25} introduced the "MS@MS-Shapley" wealth function $\intro*\msscorefun[q]$  which %
outputs, for a "monotone" query $q$,  the number $\msscorefun[q](S)$ of "minimal supports" of $q$ inside $S$.

Given a class of queries $\+C$, we denote by $\intro*\msShapley{\+C}$ ("resp" $\intro*\drShapley{\+C}$) the associated computational problem of computing, for a given "query" $q\in\+C$, "database" $\D$ and "fact" $\alpha \in \D$, the value $\Sh(\D, \msscorefun[q], \alpha)$ ("resp" $\Sh(\D, \drscorefun[q], \alpha)$), which we shall call the ""MS-Shapley"" ("resp" ""drastic-Shapley"") responsibility measures or scores.
We will naturally extend this notation $\reintro*\wShapley{\+C}{\star}$  to other "wealth function" families $\set{\xi^\star_q}_{q \in \+C}$ introduced in this paper, simply by varying the $\star$ superscript.%

\AP The "MS-Shapley" measure belongs to a larger family of responsibility measures based upon aggregating the number and sizes of "minimal supports", introduced  in \cite{ourpods25} as "weighted sums of minimal supports", or ""WSMS"". Indeed, it has been shown in \cite[Proposition~4.3]{ourpods25} that $\Sh(\D, \msscorefun[q], \alpha)$ can be equivalently and more simply defined as being 
the following sum:
\begin{align}\label{eq:msscore}
    \Sh(\D, \msscorefun[q], \alpha) &= 
    \textstyle\sum \, \set{\frac{1}{|S|} : S \text{ "minimal support" of $q$ in $\D$ s.t. } \alpha \in S }.
\end{align}
Other "WSMS" measures are obtained by %
replacing $\frac{1}{|S|}$ with any positive ""weight function"" $f(|S|)$. %
Every such WSMS measure can be equivalently defined as the Shapley value of some suitably chosen wealth function \cite[Proposition 4.4]{ourpods25}. 

For the sake of readability, we shall phrase our definitions and results using the "drastic@drastic-Shapley" and "MS-Shapley" responsibility measures, but we emphasize that this is merely to enhance readability, as all of our results for "MS-Shapley" measures extend to any "WSMS" measure based on a tractable weight function $f$, see \Cref{sec:conclusion}
\ifcameraready%
and \cite[§E.1]{thispaper-arxiv}
\else%
and \Cref{app:wsms}
\fi
for more details. 

In the present work, we will build upon the following known result ensuring tractability of $\msShapley{\+C}$ and other WSMS measures in "combined complexity".
\ifcameraready
    \begin{theorem}[{\cite[Theorem 6.6 with Lemma 5.1]
        {ourpods25}}]\label{Theorem 6.6 and more from PODS25}
        For any class $\+C$ of "CQs" having bounded "generalized hypertree width" and bounded "self-join width", the problem of counting the number of "minimal supports" of a given size is in polynomial time. Further, $\msShapley{\+C}$ is in polynomial time in "combined complexity". 
        These results extend also to "CQneq" queries and other WSMS measures based upon tractable weight functions.
    \end{theorem}
\else
    \begin{theoremrep}[{\cite[Theorem 6.6 with Lemma 5.1]
        {ourpods25}}]\label{Theorem 6.6 and more from PODS25}
        For any class $\+C$ of "CQs" having bounded "generalized hypertree width" and bounded "self-join width", the problem of counting the number of "minimal supports" of a given size is in polynomial time. Further, $\msShapley{\+C}$ is in polynomial time in "combined complexity". 
        These results extend also to "CQneq" queries and other WSMS measures based upon tractable weight functions.
    \end{theoremrep}
    \begin{proof}
        \cite[Theorem 6.6]{ourpods25} shows that counting the number of "minimal supports" by size is tractable. While the statement is technically on plain "conjunctive queries", for each "inequality atom" $\alpha = (t \neq t')$ one can pre-compute the $\neq$-relation over the "active domain" and replace $\alpha$ with some plain relation $R_\alpha$.
        Further, $\msscorefun[q]$ is polynomial-time computable, via the alternative definition of summing $\frac{1}{|S|}$ over all "minimal supports" $S\subseteq \D$ of $q$ containing $\alpha$. 
        For $\msscorefun[q]$ and other tractable weight functions,
        \cite[Lemma 5.1]{ourpods25} shows that the tractability for counting "minimal supports" by size implies tractability of $\msShapley{\+C}$.
    \end{proof}
\fi

\begin{remark}[Exogenous facts]\label{rk:exogenous}
    Prior work on database responsibility measures has considered a more complex setting in which the input database is partitioned into `exogenous' and `endogenous' facts, where the exogenous facts are treated as given and the responsibility is distributed only among endogenous facts. To keep the focus on the handling of negation, we decided to keep definitions to their most basic form by working with plain (unpartitioned) databases (i.e. treating all database facts as endogenous). It is important to note that while tractability results obtained for partitioned databases also apply to the simpler setting without exogenous facts, this is not the case for intractability results. 
    Indeed, all hardness results concerning UCQs that have been proven for "drastic-Shapley" computation crucially rely upon the presence of exogenous facts, and it is an interesting open problem whether %
    there exists a "CQ" $q$ for which $\drShapley{q}$ is "shP"-hard in the purely endogenous setting. 
\end{remark}

%% file: sec-signed.tex
\section{Responsibility Measures via Signed Supports}\ifcameraready\nosectionappendix\fi
\label{sec:signed}
In this section, we formalize the `signed fact approach' sketched in the introduction and 
analyze the data and combined complexity of the resulting responsibility measures. 

\ifcameraready
    \begin{toappendix}
        \section{Selected Omitted Proofs}
        The remaining proofs not included here can be found in the long version \cite{thispaper-arxiv}. 
    \end{toappendix}
\fi

\subsection{Relevance and Responsibility of Signed Facts}
\AP
Given a "schema" $\Sigma$, let $\AP\intro*\Sigmapm$ be the "schema" having $\intro*\spos R$ and $\intro*\sneg R$ for every $R$ in $\Sigma$ with the same "arity" as $R$. For every "atom" $R(\bar t)$ in the schema $\Sigma$, let $\spos R(\bar t)$ and $\sneg R(\bar t)$ denote the corresponding atoms over the "schema" $\Sigmapm$, which we shall call ""positive@positive atom"" and ""negative atoms"" (not to confuse with "negated atoms").\footnote{In particular \AP""positive@positive fact"" and ""negative facts"" are "facts" over $\spos$ and $\sneg$ relations, respectively.} We generally use the term ""signed@signed fact"" to refer to such  "atoms" or "facts", and we will extend this notation to sets of facts~$S$ by letting $\reintro*\spos S \eqdef \set{\spos \alpha : \alpha \in S}$
and $\reintro*\sneg S \eqdef \set{\sneg \alpha : \alpha \in S}$. 
Given a "database" $\D$ over a "schema" $\Sigma$, we let $\Dpm$ be the "database" over $\Sigmapm$ defined as follows:
\AP
\[
    \intro*\Dpm \eqdef \spos \D \cup \set{\sneg R(\bar c) : R \in \Sigma, \arity(R)=|\bar c|, \const(\bar c) \subseteq \const(\D) \text{, and } R(\bar c) \not\in \D}.
\]
\AP
For any set $S$ of "signed facts", let $\intro*\Dpos[S]$ and $\intro*\Dneg[S]$ be the sets of (unsigned) "facts" which appear "positively@positive fact" and "negatively@negative fact" in $S$, respectively, where we shall write $\Dneg$ instead of $\Dneg[(\Dpm)]$ for brevity (note that $\Dpos[(\Dpm)]$ is simply $\D$).  
Observe that all constants occurring in %
$\Dpm$ %
belong to the "active domain" of $\D$, that is, we view negation with a `closed-world' point of view, and in line with the restriction on "safe negations".

Given a "Boolean" "UCQneg"  query $q$ over $\Sigma$, let $\AP\intro*\querypm$ be the "UCQneq" over $\Sigmapm$ resulting from replacing in $q$ each "negated atom" $\lnot R(\bar t)$ with the (positive) "atom" $\sneg R(\bar t)$ over $\Sigmapm$, and each "positive atom" $R(\bar t)$ with $\spos R(\bar t)$.\footnote{Recall that $q$ may also have inequalities $x \neq y$, in which case $\querypm$ will also contain them.} Observe that
\[
    \D \models q
        \quad \Leftrightarrow \quad 
    \Dpm \models \querypm
\]
This naturally suggests the following notion of support for signed facts: 
a set $S  \subseteq \Dpm$ of "signed facts" is a \AP""signed support"" of $q$ in $\D$ if $S \models \querypm$, which is equivalent to requiring the set $S$ to be a "support" of the monotone query $\querypm$ in $\Dpm$. We argue that %
a "minimal" "signed support" is a sensible definition for a qualitative explanation for why $q$ holds in $\D$, since it contains the minimal information on which facts should be present and absent to make the query true in the context of $\D$. This in turn leads us to define the corresponding notion of relevance:
a "signed fact" is \AP""signed-relevant"" for $q$ in $\D$ if it belongs to a "minimal" "signed support" of $q$ in $\D$. We illustrate these notions on an example:

\begin{example}\label{ex:signed-example}
Consider the "CQneg" $q$ of \Cref{fig:signed-example} over directed graphs (with edge relation $E$), and the database $\D$ given by the depicted graph. Under the signed facts approach, there are two possible explanations for $q$ given by the "minimal" "signed supports"  $S_1=\{\spos E(a,b), \spos E(b,c), \sneg E(c,a)\}$ and $S_2=\{\spos E(b,c), \spos E(c,c), \sneg E(c,b) \}$. It follows that the "signed-relevant" facts are $\spos E(a,b), \spos E(b,c), \spos E(c,c), \sneg E(c,a), \sneg E(c,b)$, and the other "signed facts" in $\Dpm$, like $\spos E(b,a)$ and $\sneg E(b,b)$, would be deemed irrelevant under this approach.\qedhere
\end{example}
\begin{figure}
    \begin{center}
        \includegraphics[width=\textwidth]{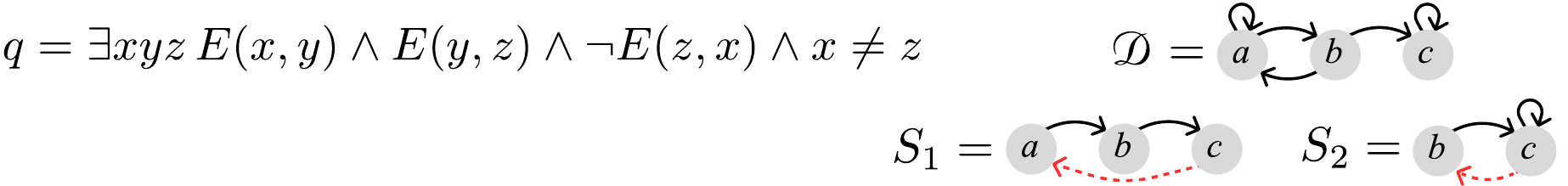}
    \end{center}
    \caption{Query and database from \Cref{ex:signed-example}, together with the "minimal" "signed supports" $S_1$ and $S_2$.  
    Solid black arrows represent "positive facts" and dotted red arrows "negative facts".}
    \label{fig:signed-example}
\end{figure}

The following lemma clarifies how "signed supports" relate to finite entailment.

\begin{lemmarep}\label{cqneg-alt-signed}
For every "UCQneg" $q$, database $\D$, and set $S  \subseteq \Dpm$, 
if $S$ is a "signed support" of $q$ in $\D$, then $\D' \models q$ for every database $\D'$ with $\Dpos[S] \subseteq \D'$ and $\D' \cap \Dneg[S] = \emptyset$ (equivalently, $\Dpos[S] \cup \{\neg R(\bar t) : R(\bar t) \in \Dneg[S]\}  \models^\mathsf{fin} q$). However, the converse does not hold even for "CQneg"s.
\end{lemmarep}
\begin{proof}
    If $S$ is a "signed support", we have $S \models \querypm$ via some "satisfying assignment" $\nu : \vars(\querypm) \to \const(S)$. Observe that the same mapping $\nu$ is a "satisfying assignment" witnessing $\D' \models q$ as soon as $\D'$ contains all "positive facts" of $S$ and no "negative fact" of $S$.

    The fact that the converse implication does not hold means that there exist a "CQneg" $q$, database $\D$, and set $S  \subseteq \Dpm$ such that: (1) $S$ is not a "signed support" and yet (2) we have that $\D' \models q$ for every database $\D'$ with $\Dpos[S] \subseteq \D'$ and $\D' \cap \Dneg[S] = \emptyset$. Such $S,\D,q$ are exhibited in \Cref{ex:alternative-entailment-semantics-signed}.
\end{proof}

To see why the converse statement in Lemma \ref{cqneg-alt-signed} fails and why an alternative definition of minimal support via finite entailment can yield unexpected results, it is instructive to consider the following example.

\begin{example}\label{ex:alternative-entailment-semantics-signed}
Consider the "CQneg" $q=\exists x,y ~ A(x) \land R(x,y) \land \lnot A(y)$ and database $\D = \{A(b), R(b,c), A(c), R(c,d)\}$. The only "minimal" "signed support" is $\set{\spos A(c), \spos R(c,d), \sneg A(d)}$. However, if we were to adopt the \AP""entailment semantics"" that defines minimal supports using finite entailment (cf.\ Lemma \ref{cqneg-alt-signed}), then there would be a further minimal support $S = \set{\spos A(b), \spos R(b,c), \spos R(c,d), \sneg A(d)}$. %
Indeed, for any database $\D'$ with $\Dpos[S] \subseteq \D'$ and $\D \cap \Dneg[S] = \emptyset$, one of two things must happen: either (i) $A(c)\in \D'$,  in which case $\D'\models q$ by $(x,y)\mapsto (c,d)$, or (ii) $A(c)\notin \D'$, in which case we still have $\D'\models q$, this time by $(x,y)\mapsto (b,c)$.
This is at odds with the fact that since $A(c) \in \D$ there is no satisfying valuation for $q$ that involves the constant $b$
and hence $A(b)$ is not useful for obtaining $q$ in the considered database $\D$.
\end{example}

In addition to yielding an arguably unnatural notion of relevance, 
the "entailment semantics" has another disadvantage, namely, it can yield  minimal supports of unbounded size:

\begin{example}\label{ex:signed-alt-ent-unbounded}
Reconsider the  "CQneg" $q=\exists x,y ~ A(x) \land R(x,y) \land \lnot A(y)$, and for each $n \geq 1$, define the database
$\D_n = \set{A(c_i) : 0 \leq i < n } \cup \{R(c_i,c_{i+1}) : 0 \leq i < n\}$ 
(note that 
$\D_n$ does not contain $A(c_n)$). 
Then 
for every $n\geq 1$, the "entailment semantics" would give rise to the minimal support 
$S_n = \set{\spos A(c_0), \sneg A(c_n)} \cup \{\spos R(c_i,c_{i+1}) : 0 \leq i < n\}$. 
\end{example}
We find it both unintuitive and computationally complex to use a notion of explanation for "UCQneg"s
that produces explanations whose number of atoms do not depend on the size of the query.
 For this reason, we will not consider this alternative notion further and instead retain our original definition of "signed support" defined in terms of minimal subsets of $\Dpm$ that satisfy (rather than entail) $q$.

With suitable notions of "signed-relevance" and "signed supports" at hand, we can now introduce responsibility measures for signed facts,
by simply applying existing responsibility measures for %
monotone queries
 to the signed version of the query ($\querypm$) and input database ($\Dpm$).  
In particular, we will be interested in the following Shapley-value-based measures:
\AP 
\begin{itemize}
\item ""MS-Shapley measure for signed facts"": $\Sh(\Dpm, \msscorefun[\querypm], \alpha)$
\item ""Drastic-Shapley measure for signed facts"": $\Sh(\Dpm, \drscorefun[\querypm], \alpha)$
\end{itemize}
Observe that since $\querypm$ is monotone, these measures always return non-negative numbers. 

Moreover, we can show that these measures allow us to decide %
 "signed-relevance":

\begin{lemmarep}\label{signed-rel-nonzero}
For every "database" $\D$, "signed fact" $\alpha \in \Dpm$, and "UCQneg" $q$: 
 $\alpha$ is "signed-relevant" for $q$ in $\D$ iff $\Sh(\Dpm, \msscorefun[\querypm], \alpha)> 0$ iff $\Sh(\Dpm, \drscorefun[\querypm], \alpha) >0$.
\end{lemmarep}
\begin{proof}
    Observe first that since $\querypm$ is "monotone", both scores are defined as sums of non-negative numbers. If $S$ is a "minimal" "signed support" of $q$ "wrt" $\D$ containing $\alpha$, then $\drscorefun[\querypm](S) = 1$, $\msscorefun[\querypm](S) = 1$, $\drscorefun[\querypm](S \setminus\set\alpha) = 0$ and $\msscorefun[\querypm](S\setminus\set\alpha) = 0$, yielding a strictly positive number both for $\Sh(\Dpm, \msscorefun[\querypm], \alpha)$ and $\Sh(\Dpm, \drscorefun[\querypm], \alpha)$.
    If, on the other hand, $\alpha$ is in no "minimal" "signed support", then for every set $S \subseteq \Dpm$ we have $\msscorefun[\querypm](S) = \msscorefun[\querypm](S \setminus \set\alpha)$ and
    $\drscorefun[\querypm](S) = \drscorefun[\querypm](S\setminus\set\alpha)$. This in turn implies that $\Sh(\Dpm, \msscorefun[\querypm], \alpha) = \Sh(\Dpm, \drscorefun[\querypm], \alpha) = 0$.
\end{proof}

\begin{example}\label{ex:signed-example2}
Reconsider the "CQneg" $q$ and database $\D$ from \Cref{fig:signed-example}.  
Under the "MS-Shapley" measure for signed facts, the positive fact $\spos E(b, c)$ would get value $2/3 = 1/3 + 1/3$ according to \Cref{eq:msscore} since it participates in two "minimal" "signed supports", both of size 3. Each of the positive facts $\spos E(a, b)$ and $\spos  E(c,c)$ and the negative facts $\sneg E(c,a)$ and $\sneg E(c,b)$ receives a value $1/3$,
corresponding to a participation in a single "minimal" "signed support" of size 3.  %
The remaining "signed facts" (which are not "signed-relevant") all receive a null score.\qedhere
\end{example}

The remainder of this section will be devoted to studying the complexity of computing responsibility scores for "signed facts". %
Formally, 
for a class $\+C$ of "Boolean" "UCQneg" queries, we denote by 
\AP
$\intro*\ShapleypmMinSup{\+C}$ the problem of computing $\Sh(\Dpm, \msscorefun[\querypm], \alpha)$ given $\D$, $\alpha \in \Dpm$ and $q \in \+C$;
and similarly for $\intro*\ShapleypmDrastic{\+C}$ with $\drscorefun[\querypm]$.
We shall investigate these problems both in ""data@data complexity"" and ""combined complexity"", depending on whether the query is considered to be fixed or not.

We first observe that we can restrict to having only "negative facts" of relations which appear under a negation in the query.
Let \AP$\intro*\NegRelsq$ be the set of relations %
    that appear in "negated atoms" in $q$, and let $\AP\intro*\Dpmq \subseteq \Dpm$ be $\Dpm \setminus \set{\sneg R(\bar c) : R \not\in \NegRelsq}$.
\begin{lemmarep}\label{lem:Dpm-restriction-to-negative-atoms}
     Every "minimal" "signed support" of $\D,q$ is in $\Dpmq$.
\end{lemmarep}
\begin{proof}
    By means of contradiction, suppose a "signed support" $S$ of $q$ contains some $\alpha = \sneg R(\bar c)$ for some $R \not\in \NegRelsq$. Let $\nu$ be a "satisfying assignment" of $\querypm$ in $S$. Note that $\nu$ is still a "satisfying assignment" of $\querypm$ in $S \setminus \set{\alpha}$ since $\querypm$ contains no atom over the `$\sneg R$' relation of $\Sigmapm$. Hence, $S$ cannot be a "minimal" "signed support".
\end{proof}

We say that a class $\+C$ of "UCQneg" has \AP""bounded negative arity"" if there is a bound $N$ such that $\arity(R) \leq N$ for all $q \in \+C$ and $R \in \NegRelsq$. The following lemma formalizes a simple and useful (many-one) reduction to queries without negation:
\begin{lemma}\label{lem:redux:mssigned-to-msnormal}
    For any class $\+C \subseteq \text{"UCQneg"}$ of "bounded negative arity", there are polynomial-time reductions from $\ShapleypmMinSup{\+C}$ to $\msShapley{\+C^\pm}$ and from 
    $\ShapleypmDrastic{\+C}$ to  $\drShapley{\+C^\pm}$, where $\+C^\pm = \set{\querypm : q \in \+C}$.
\end{lemma}
\begin{proof}
    Given a "UCQneg" $q \in \+C$, a "database" $\+D$, and "signed fact" $\alpha \in \Dpm$, we need to compute $\Sh(\Dpm, \msscorefun[\querypm], \alpha)$. First observe that since $\+C$ is assumed to have "bounded negative arity", we can build $\Dpmq$ in polynomial time.
    Moreover, by \Cref{lem:Dpm-restriction-to-negative-atoms}, we can restrict our attention to $\Dpmq$ instead of $\Dpm$, "ie",
    $\Sh(\Dpm, \msscorefun[\querypm], \alpha) = \Sh(\Dpmq, \msscorefun[\querypm], \alpha)$ and $\Sh(\Dpm, \drscorefun[\querypm], \alpha) = \Sh(\Dpmq, \drscorefun[\querypm], \alpha)$, to obtain the desired many-one reduction to the instance $(\Dpmq,\querypm,\alpha)$ of $\msShapley{\+C^\pm}$ /  $\drShapley{\+C^\pm}$.
\end{proof}

\subsection{Computing MS-Shapley Values for Signed Facts}\label{ss:ms-signed}

We shall now study the complexity of computing "MS-Shapley" values for "signed facts". 

\subparagraph*{Data Complexity}
Regarding data complexity, it was shown in \cite{ourpods25} that for every "UCQ" $q$, the problem %
$\msShapley{q}$ is tractable. We first observe that this result immediately lifts to "UCQneg", by applying the reduction in Lemma \ref{lem:redux:mssigned-to-msnormal}. 

\begin{propositionrep}\label{prop:UCQneg-poly-MS-data:signed}
    $\ShapleypmMinSup{\text{"UCQneg"}}$ is in polynomial time in data complexity.
\end{propositionrep}
\begin{proof}
    Given $q$, $\D$, $\alpha$, first observe that  since we work in "data complexity", the arity of "negated atoms" in $q$ is fixed.
    We can then apply \Cref{lem:redux:mssigned-to-msnormal} to reduce to $\msShapley{\text{"UCQneq"}}$. By \cite[Theorem 5.2]{ourpods25}, and using the facts  that $\msscorefun[]$ is tractable and that "UCQneq"s are bounded, monotone and tractable,
    the "data complexity" of $\msShapley{\text{"UCQneq"}}$ is in polynomial time.
\end{proof}

As an immediate corollary of the preceding proposition and Lemma \ref{signed-rel-nonzero}, 
we can further conclude that deciding "signed-relevance" is tractable. In fact, with a bit more work, 
we can show that this problem enjoys the lowest possible data complexity (namely, "AC0"). 

\begin{proposition}\label{prop:signed-relevance-AC0}
It can be decided in "AC0" data complexity whether a "signed fact" $\alpha \in \Dpm$ is "signed-relevant" w.r.t.\ a given "UCQneg" and database $\D$. 
\end{proposition}
\begin{proof}
    Consider a  "UCQneg" $q$. We will write ``$p \in \querypm$'' to denote that $p$ is a "CQneg" disjunct of~$\querypm$. 
    To test whether a given "signed fact" $\alpha$ is "signed-relevant" w.r.t.\ $q$ and the input database, 
    we want to express the following three-part condition
    \begin{enumerate}
        \item \label{prop:signed-relevance-AC0:1} 
        There is a "signed support" (for some disjunct $p \in \querypm$)...
        \item \label{prop:signed-relevance-AC0:2} 
        ...that contains $\alpha$, such that...
        \item \label{prop:signed-relevance-AC0:3} 
        ...it does not strictly contain a "signed support" (for some disjunct $p' \in \querypm$).
    \end{enumerate}
 by means of a first-order sentence that can be evaluated on $\D$ (as opposed to $\Dpm$), and which can be constructed from $q$ alone.

    Let $\AP\intro*\Atsubp p$ be the set of all relational "atoms" in a "CQneg" $p \in\querypm$.
    We will write $\beta = \beta'$ for a "signed fact" $\beta$  and a signed atom $\beta'$ (containing only variables as terms) to denote that the relation name and sign are the same, and moreover $\bigwedge_i x_i = c_i$ holds, where $i$ is the arity of their shared relation and $x_i$ (resp.\ $c_i$) is the $i$th variable (resp.\ constant) %
    of $\beta'$ (resp.\ $\beta$). %
    For $\mu : \vars(q) \to \vars(q)$ and $p \in \querypm$, let
    $\AP\intro*\psubmu p \mu$ denote the Boolean query $\hat p \land \bigwedge_{x,y \in Im(\mu)} x \neq y$, where $\hat p$ is obtained by replacing each variable $x$ with $\mu(x)$ in the disjunct $\check p$ of $q$ corresponding to $p$ (note that both $\check p$ and $\hat p$ may contain "negated atoms"). 
    For a set $At$ of signed atoms, let $\mu(At)$ denote the result of replacing each variable $x$  in $At$ with $\mu(x)$.

    The above condition can then be captured with the following first-order sentence:    \[
    \underbrace{\bigvee_{p \in \querypm} \bigvee_{\mu: \vars(p) \to \vars(q)}  \Big( \psubmu p \mu}_{\text{\Cref{prop:signed-relevance-AC0:1}}} {} \land {}
    \underbrace{\bigvee_{\beta \in \Atsubp p} \alpha = \mu(\beta)}_{\text{\Cref{prop:signed-relevance-AC0:2}}}
    {} \land {}
    \underbrace{\lnot \bigvee_{p' \in \querypm} \bigvee_{\mu' : \vars(p') \to \vars(q) \text{ s.t. } \mu'(\Atsubp{p'}) \subsetneq \mu(\Atsubp p)} \psubmu{p'}{\mu'} \Big)}_{\text{\Cref{prop:signed-relevance-AC0:3}}}
    \]
    Observe that this formula holds in $\D$ "iff" $\alpha$ is "signed-relevant" "wrt" $q$ and $\D$. 
    Since evaluating first-order formulas is in "AC0" for data complexity, the statement follows.
\end{proof}

\subparagraph*{Combined Complexity}
It is known from prior work that %
the combined complexity of computing "MS-Shapley" scores for "UCQs" without "negated atoms" is "shP"-hard, even under severe syntactic restrictions (see \cite[Proposition 13]{ourKR25}). This is why we focus %
on "CQneg"s.

Enumerating all "minimal" "signed supports" for $q$ in $\D$ as done in \Cref{prop:UCQneg-poly-MS-data:signed} is not an option for obtaining tractable "combined complexity" since the number of supports can be %
of the order of $O(|\D|^{|q|})$.
The recent work \cite{ourpods25}, rephrased in \Cref{Theorem 6.6 and more from PODS25}, provides  conditions for tractability of counting "minimal supports" for "CQ"s. %
Since counting "minimal" "signed supports" for a "CQneg" $q$ corresponds to counting "minimal supports" of $\querypm$ on $\Dpm$, we can obtain a similar condition for classes $\+C$ of "CQneg"s  ensuring tractability of $\ShapleypmMinSup{\+C}$ by reusing the same algorithm developed in \cite[Theorem 6.6]{ourpods25}.
The only potential issue is that the "database" $\Dpm$ which we need to consider can be substantially larger than $\D$, and for this reason, we need to bound the arity of the "negated atoms" of $q$ %
to ensure $\Dpm$ remains polynomial in size. 

We say that a class $\+C$ of "CQneg" has bounded "generalized hypertree width" (or bounded "self-join width") if the corresponding class of "CQneq"s $\{\querypm : q \in \+C\}$ enjoys this property. We are now ready to state our tractability result:

\begin{propositionrep}\label{prop:signed-ms}
    For every class $\+C$ of "CQneg"  queries having "bounded negative arity", bounded "generalized hypertree width" and bounded "self-join width",
    $\ShapleypmMinSup{\+C}$ is in polynomial time in "combined complexity", as a corollary of \Cref{Theorem 6.6 and more from PODS25,lem:redux:mssigned-to-msnormal}.
\end{propositionrep}
\begin{proof}
    By \Cref{lem:redux:mssigned-to-msnormal} we can reduce to $\msShapley{\+C^\pm}$, where $\+C^\pm = \set{\querypm : q \in \+C} \subseteq \text{"CQneq"}$. Since the "self-join width" and "generalized hypertree width" of $\+C$ and $\+C^\pm$ coincide by definition, we can conclude by \Cref{Theorem 6.6 and more from PODS25}.
\end{proof}

\subsection{Computing Drastic-Shapley Values for Signed Facts}\label{ss:dr-signed}
We now turn to the problem of computing drastic-Shapley values for "signed facts". 
From the existing literature on "drastic-Shapley"-based measures for "monotone queries", %
the largest known class of "UCQ"s for which SVC is tractable in data complexity 
is the class of ""safe UCQs"" \cite[Corollary 3.2]{deutchComputingShapleyValue2022a}.
We refer readers to \cite[Definition 4.10]{dalviDichotomyProbabilisticInference2012} for a formal definition but 
point out that the "safe UCQ"s have been shown to characterize the "UCQ"s which can be tractably evaluated on probabilistic databases (assuming $\FP\neq\sP$).
The following result shows that this tractability result can be lifted to "UCQneg"s. In the statement, %
we shall say that a "UCQneg" is ""safe@@UCQneg"" when  $\querypm$ is a  "safe UCQ", where the %
`$\neq$' relation is treated like any binary relation.

\begin{proposition}\label{prop:signed-dr-tractability}
For every "safe@@UCQneg" "UCQneg" query $q$, $\ShapleypmDrastic{q} \in \FP$. 
\end{proposition}
\begin{proof}
    By \Cref{lem:redux:mssigned-to-msnormal} we can reduce to $\drShapley{\querypm}$. 
    In order to compute $\Sh(\D, \drscorefun[\querypm], \alpha)$, it suffices to first materialize $\D_{\neq} = \D \cup \{{\neq}(c,d) : c,d \in \const(\D) \cup \const(q), c \neq d \}$, which we can do in polynomial time,
    and then to compute  $\Sh(\D_{\neq}, \drscorefun[\querypm], \alpha)$ using the materialized $\neq$-relation to evaluate "inequality atoms" in $\querypm$. As $q$ is assumed to be a "safe@@UCQneg" "UCQneg", it follows that $\querypm$ is a "safe UCQ", so the latter task can be performed in $\FP$ due to \cite[Corollary 3.2]{deutchComputingShapleyValue2022a}. 
\end{proof}

While some "shP"-hardness results and $\FP$-"shP" dichotomies have been proven for various subclasses of "UCQ"s \cite{livshitsShapleyValueTuples2021,ourpods24}, all of the existing hardness proofs %
crucially rely on the so-called ``exogenous'' facts, which we do not consider in this work ("cf" \Cref{rk:exogenous}). 
Moreover, even if new hardness results were to be proven for the purely endogenous setting, 
 it is not at all obvious how one would reduce $\drShapley{\querypm}$
to $\ShapleypmDrastic{q}$. This is because while the former problem considers arbitrary databases over the signed schema $\Sigmapm$,
the problem $\ShapleypmDrastic{\querypm}$ is defined for input databases of the restricted form $\Dpm$.

%% file: sec-global.tex
\section{Responsibility Measures via Positive Supports}\ifcameraready\nosectionappendix\fi
\label{sec:global-monotone}
The "signed facts" approach put forth in the previous section may not always be suitable, %
either because one may prefer to attribute responsibility to more tangible elements which are explicitly present in the database instead of `watering down' the responsibility score among `absent' facts, or simply because the number of "negative facts" is just too large, especially if one considers "schemas" with relations of high arity (in particular, note that \Cref{prop:signed-ms} assumes "bounded negative arity"). This motivates us to explore how to define notions of relevance and responsibility measures for positive facts. 

\subsection{Relevance and Responsibility of Positive Facts}

One needs to be careful when defining supports based solely on positive facts. %
As already discussed in the introduction in relation to the example in \Cref{fig:ex-recipe}, we should avoid considering $\set{I(\textit{mp},\textit{fish})}$ as a valid explanation, because the satisfaction of $q$ crucially relies on the absence of $I(\textit{mp},\textit{meat})$, which is not justified in the context of the database $\D$, which contains the latter fact.
With this in mind, in order to attribute responsibility only to (positive) database facts, we shall consider that a set $S \subseteq \D$ is a \AP""positive support"" for $q$ in $\D$ if $\spos S \cup \sneg \Dneg[\D] \models \querypm$. 
Observe how the database context is still represented by $\Dneg[\D]$: we cannot assume the negation of a fact just because it is not in $S$, it must be also not in $\D$.
\AP We say that a database fact is ""positive-relevant"" if it appears in some "minimal" "positive support". 

These "positive supports" are `monotone' in the following sense. Let us call a set $S$ of "facts" a \AP""$\D$-monotone support"" for a query $q$ if $S \subseteq \D$ and further $S' \models q$ for every $S \subseteq S' \subseteq \D$.
\ifcameraready
   \begin{lemma}\label{lem:possup->DmonSup}
      If $S$ is a "positive support" of $q$ in $\D$, then $S$ is a "$\D$-monotone support".
   \end{lemma}
\else
   \begin{lemmarep}\label{lem:possup->DmonSup}
      If $S$ is a "positive support" of $q$ in $\D$, then $S$ is a "$\D$-monotone support".
   \end{lemmarep}
   \begin{proof}
      Since $\querypm$ is monotone (it is a "UCQneq"), if $\spos S \cup \sneg \Dneg[\D] \models \querypm$ and $S \subseteq S'$, then $\spos S' \cup \sneg \Dneg[\D] \models \querypm$.
   \end{proof}
\fi
However, "$\D$-monotone supports" do not in general coincide with "positive supports", nor do "minimal" "$\D$-monotone supports" yield a good notion of relevance, as the next examples show.

\begin{example}\label{ex:Dmon-/->PosSup:1}
   Consider the query $q=(\exists x ~ A(x) \land \lnot B(x)) \lor (\exists x ~ B(x) \land C(x))$ and $\D = \set{A(c), B(c), C(c)}$ from \Cref{ex:alternative-entailment-semantics-signed}. We have that $\set{A(c),C(c)}$ is a ("minimal") "$\D$-monotone support" since $\set{A(c),C(c)} \models q$ and $\set{A(c),B(c),C(c)} \models q$, but it is not a "positive support". 
   As discussed before, $\set{A(c),C(c)}$ should not be deemed a valid explanation for $q$ in the context of $\D$. Thus, we argue that the seemingly simpler notion of considering supports which are monotone does not yield a reasonable notion of explanation. 
\end{example}

Observe that the previous example uses disjunction in a non-trivial way, as $\set{A(c)}$ and $\set{A(c),B(c)}$ satisfy the query due to different disjuncts of $q$ being made true. A natural question is then whether "$\D$-monotone supports" and "positive supports" coincide for "CQneg" queries. The answer, again, is negative.
\begin{example}\label{ex:Dmon-/->PosSup:2}
   Consider the "CQneg"  $q = \exists xyzu ~ R(x,y,y) \land R(y,z,u) \land \lnot R(u,x,x)$, the "database" $\D = \set{R(a,b,b), R(b,c,d), R(d,a,a)}$, and the set of "facts" $S = \set{R(a,b,b), R(b,c,d)}$. 
   First observe  that $S$ is a "$\D$-monotone support" since we have $S \models q$ via the "satisfying assignment" $\nu = (x,y,z,u) \mapsto (a,b,c,d)$, and $\D \models q$ via $ \nu' = (x,y,z,u) \mapsto (d,a,b,b)$.
   Moreover, $S$ is a "subset-minimal@minimal" "$\D$-monotone support", since removing any set of facts from $S$ results in a non-"support".
   
   On the other hand, %
   $S$ is not a "positive support" since the sole "satisfying assignment" for the positive part of $q$ is $\nu$, sending the negated atom $R(u,x,x)$ to an existing atom $R(d,a,a) \in \D$. 
   Further, the only "minimal" "positive support" is $S' = \set{R(b,c,d), R(d,a,a)}$.
   For this reason, we would obtain that $R(a,b,b)$ is `monotone-relevant' in the sense that it belongs to some "minimal" "$\D$-monotone support" while it is not "positive-relevant".
   Analogously as in previous examples, we consider that any notion making $R(a,b,b)$ relevant is not a sensible choice for responsibility attribution, since $R(a,b,b)$ does not effectively participate in any "satisfying assignment" making the query $q$ true in $\D$.
\end{example}

An additional drawback of the notion of "$\D$-monotone supports" is that it can lead to "minimal" supports of unbounded size.

\begin{figure}
   \begin{center}
   \includegraphics[width=.9\textwidth]{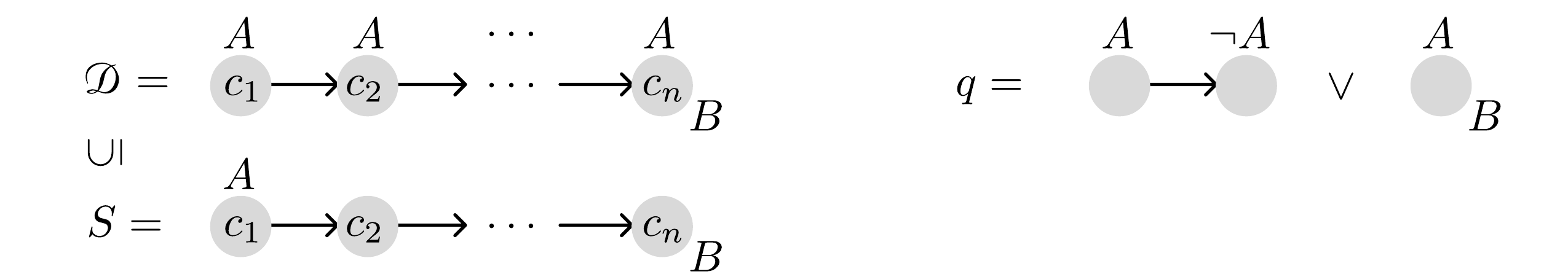}
   \end{center}
   \caption{Construction of "minimal" "$\D$-monotone supports" of unbounded size. %
   }
   \label{fig:unbounded-minimal-supports}
\end{figure}
\begin{example}
\label{ex:unbounded-minimal-supports}
   Consider, as depicted in \Cref{fig:unbounded-minimal-supports}, the database $\D_n = \set{R(c_i,c_{i+1}) : i < n} \cup \set{A(c_i) : i \leq n} \cup \set{B(c_n)}$ and the "UCQneg" 
$
         q = 
            (\exists xy ~ A(x) \land R(x,y) \land \lnot A(y)) 
         ~~\lor~~ 
            (\exists x ~ A(x) \land B(x)),
$
   adapted from \cite{DBLP:conf/kr/BienvenuB23}. 
    We can observe that $S = \set{R(c_i,c_{i+1}) : i < n} \cup \set{A(c_1)} \cup \set{B(c_n)}$ is a "minimal" "$\D$-monotone support". 
\end{example}

Summing up what we have learned from the preceding examples: %
\begin{lemmarep}
   Every "positive support" of $q$ in $\D$ is "$\D$-monotone@$\D$-monotone support" for $q$, but not all "$\D$-monotone supports" are "positive supports", even when restricted to $q$ being a "CQneg". Further, "$\D$-monotone supports" may be of unbounded size w.r.t.\ the size of the query. 
\end{lemmarep}
\begin{proof}
   The first statement is \Cref{lem:possup->DmonSup}, the second statement is given by \Cref{ex:Dmon-/->PosSup:1,ex:Dmon-/->PosSup:2}, and the third statement by \Cref{ex:unbounded-minimal-supports}. 
   
For the third statement, we provide additional details here. Reconsider the database $\D_n = \set{R(c_i,c_{i+1}) : i < n} \cup \set{A(c_i) : i \leq n} \cup \set{B(c_n)}$ and the "UCQneg" 
\ifcameraready
   $q = 
            (\exists xy ~ A(x) \land R(x,y) \land \lnot A(x)) 
         \lor
            (\exists x ~ A(x) \land B(x))$,
\else
   \[
         q = 
            (\exists xy ~ A(x) \land R(x,y) \land \lnot A(x)) 
         ~~\lor~~ 
            (\exists x ~ A(x) \land B(x)),
      \] 
\fi
      which is depicted in \Cref{fig:unbounded-minimal-supports}.
    We observe that $S = \set{R(c_i,c_{i+1}) : i < n} \cup \set{A(c_1)} \cup \set{B(c_n)}$ is a "minimal" "$\D$-monotone support". Concretely:
   \begin{enumerate}[(a)]
      \item $S \models q$, via the assignment $\set{x \mapsto c_1, y \mapsto c_2}$ which makes the first disjunct hold,
      \item $S' \models q$ for all $S \subseteq S' \subseteq \D$, since there is always a `$A \xrightarrow{R} \lnot A$' pattern unless all $A$-facts are present, in which case $\exists x ~ A(x) \land B(x)$ holds (and thus so does $q$), and
      \item $S$ is "subset-minimal@minimal" with respect to (a) and (b), since:
      \ifcameraready
            (i) removing any set of facts containing $A(c_1)$ leads to a set of facts which does not satisfy $q$,
            (ii) removing any set of facts containing $B(c_n)$ leads to a set $S'$ of facts which is not monotone since $S' \cup \set{A(c_j) : 1 \leq j \leq n} \not\models q$,
            (iii) removing any set of facts containing $R(c_i,c_{i+1})$ leads to a set $S'$ of facts which is not monotone since $S' \cup \set{A(c_j) : 1 \leq j \leq i} \not\models q$.\qedhere
      \else
         \begin{enumerate}[(i)]
            \item 
            removing any set of facts containing $A(c_1)$ leads to a set of facts which does not satisfy $q$,
            \item 
            removing any set of facts containing $B(c_n)$ leads to a set $S'$ of facts which is not monotone since $S' \cup \set{A(c_j) : 1 \leq j \leq n} \not\models q$,
            \item 
            removing any set of facts containing $R(c_i,c_{i+1})$ leads to a set $S'$ of facts which is not monotone since $S' \cup \set{A(c_j) : 1 \leq j \leq i} \not\models q$.\qedhere
         \end{enumerate}
      \fi
      
   \end{enumerate}
\end{proof}

In light of the preceding examples and results, we shall not study the notion of "$\D$-monotone supports" further and will focus instead on "positive supports" for the remainder of the section. 
Observe that the definition of "positive supports" has an obvious relation to that of "signed support", and it can be seen that every "minimal" "positive support" must be part of the positive part of a "signed support". However, the notion of relevance over (positive) facts is in general more restrictive than that of "signed-relevance", as we show next.

\begin{lemmarep}\label{lem:mon-relevance-vs-signed-relevance}
   Every "positive-relevant" fact is "signed-relevant" as a "positive fact" for all "UCQneg"s. However, not every "signed-relevant" "positive fact" is "positive-relevant", even for "CQneg"s.
\end{lemmarep}
\begin{proofsketch}
For the second statement, we consider the database $\D = \set{R(a,b), R(a,c), B(b)}$ and the "CQneg" $q = \exists x y z ~ R(x,y) \land R(x,z) \land \lnot A(y) \land \lnot B(z)$. It is easy to see that $S = \set{\spos R(a,b), \spos R(a,c), \sneg A(b), \sneg B(c)}$ is a "minimal" "signed support" and thus that $\spos R(a,b)$ is "signed-relevant". However, the only "minimal" "positive support" is $S' = \set{R(a,c)}$, meaning that $R(a,b)$ is not "positive-relevant". Indeed, observe that $\set{R(a,b)}$ is not a "positive support" since it would require having $\sneg B(b)$ in $\sneg \Dneg$.
\end{proofsketch}
\begin{proof}
   \proofcase{First statement.}
   Consider a database $\D$ and "UCQneg" $q$, and let $\alpha \in \D$ be "positive-relevant". It follows that there exists a "minimal" "positive support" $S$ such that $\alpha \in S$. In particular, $\spos S \cup \sneg\Dneg \models \querypm$, and hence $\spos S \cup \sneg \Dneg$ is a "signed support" of $q$.
 If $\spos S \cup \sneg \Dneg$ is a "minimal" "signed support", then we are done, as $\spos S \cup \sneg \Dneg$  witnesses that $\spos \alpha$ is "signed-relevant".
 Otherwise, take any $S' \subseteq S$ and $D' \subseteq \Dneg$ such that $+S' \cup \sneg D'$ is a "minimal" "signed support". 
 We must have $S'=S$, otherwise $S'$ would also be a "positive support", contradicting the minimality of $S$.
 Hence, $+S \cup \sneg D'$ is a "minimal" "signed support" that contains $+\alpha$, so $\spos \alpha$ is "signed-relevant".
   \medskip 

   \proofcase{Second statement.} Consider the database $\D = \set{R(a,b), R(a,c), B(b)}$ and the "CQneg" $q = \exists x y z ~ R(x,y) \land R(x,z) \land \lnot A(y) \land \lnot B(z)$, see \Cref{fig:mon-relevance-vs-signed-relevance} for a depiction.
   \begin{figure}
      \includegraphics[width=\textwidth]{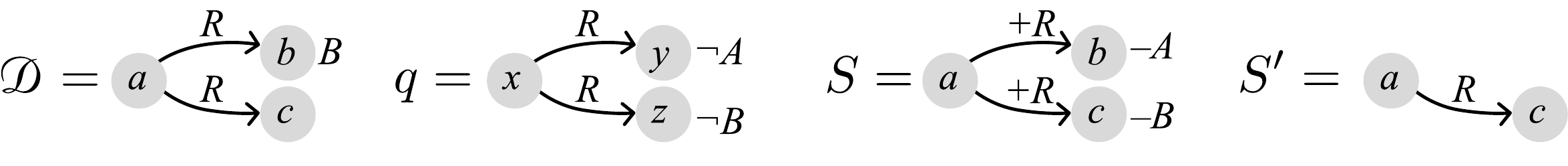}
      \caption{Counterexample for proof of \Cref{lem:mon-relevance-vs-signed-relevance}.}
      \label{fig:mon-relevance-vs-signed-relevance}
   \end{figure}
   It is easy to see that $S = \set{\spos R(a,b), \spos R(a,c), \sneg A(b), \sneg B(c)}$ is a "minimal" "signed support" and thus that $\spos R(a,b)$ is "signed-relevant". However, the only "minimal" "positive support" is $S' = \set{R(a,c)}$, meaning that $R(a,b)$ is not "positive-relevant". Indeed, observe that $\set{R(a,b)}$ is not a "positive support" since it would require having $\sneg B(b)$ in $\sneg \Dneg$.
\end{proof}

However, when a "CQneg" has no "mergeable atoms" (which holds, in particular, if it is "self-join free"), then the two notions coincide.
\begin{lemmarep}\label{lem:no-mergeable-atoms-implies-bijection}
   Let $q$ be a "CQneg" with no "mergeable atoms", $\D$ a database and $S \subseteq \D$. Then, $S$ is a "minimal" "positive support" if, and only if, $\hat S$ is a "minimal" "signed support" for some $\hat S$ such that $\Dpos[\hat S] = S$. Further, if $S$ is a "minimal" "positive support", then there exists exactly one "minimal" "signed support" $\hat S$ such that $\Dpos[\hat S] = S$.
\end{lemmarep}
\begin{proof}
   First observe that:
   \begin{enumerate}
      \item \label{it:atleast-pos-atoms}
      If $\hat S \models \querypm$, then the size of  $\Dpos[\hat S]$ must be, at least, the number $k$ of "positive atoms" of $q$, since otherwise we would have two "mergeable atoms" in $q$.
      \item \label{it:positive-determines-negative}
      If $\hat S$ is a "minimal" "signed support", then there exists exactly one "homomorphism" $\querypm \homto \hat S$. Indeed, there is only one homomorphism from the "positive atoms" to the "positive facts" since otherwise there would be two "mergeable atoms"; and the variable assignment of the unique homomorphism for the "positive atoms" also fixes where the variables in the "negative atoms" are mapped since all negations in $q$ are "safe@safe negation".
   \end{enumerate}

   \proofcase{First statement.} 
   The left-to-right direction is shown in the proof of the first statement of \Cref{lem:mon-relevance-vs-signed-relevance}. 
   For the right-to-left direction, 
    suppose now that $\hat S \subseteq \Dpm$ is a "minimal" "signed support" -- in particular $\hat S \models \querypm$. Note that by "minimality" the size of $\Dpos[\hat S]$ must be exactly $k$ %
    by \Cref{it:atleast-pos-atoms}.
   By "monotonicity" of $\querypm$ (it is a "CQneq"), we have $\spos \Dpos[\hat S] \cup \sneg \Dneg \models \querypm$ and thus $\Dpos[\hat S]$ is a "positive support". If $\Dpos[\hat S]$ was not "minimal", we would have $\spos S' \cup \sneg \Dneg \models \querypm$ for some $S' \subsetneq \Dpos[\hat S]$, contradicting that the number of "positive atoms" of $\spos S' \cup \sneg \Dneg$ is at least $k$. %

   \medskip    

   \proofcase{Second statement.}
   For any given "minimal" "signed support" $\hat S$, there is only one homomorphism from the positive atoms of $\querypm$ to $\spos \Dpos[\hat S]$, as explained in \Cref{it:positive-determines-negative}, and this homomorphism fixes the "negative facts" of $\hat S$. Hence, there cannot be two "minimal" "signed supports" sharing the same "positive atoms".
\end{proof}
\begin{corollary}%
   For any "CQneg" with no "mergeable atoms", a fact $\alpha$ is "positive-relevant" "iff" $\spos\alpha$ is "signed-relevant".
\end{corollary}

\AP
Utilizing the notion of "positive supports", we may define Shapley-value-based responsibility measures analogously to what has been done for the class of all "supports" for "monotone" queries in \cite{ourpods25}. 
Concretely, for a database $\D$ and Boolean query $q$, we define the ""positive-drastic-Shapley"" value of a fact $\alpha \in \D$ as being $\Sh(\D,\monDRscorefun[q],\alpha)$, where for every $S\subseteq \D$, we define
 $\intro*\monDRscorefun[q](S)=1$ if $S$ contains a "positive support" for $q$ and $\monDRscorefun[q](S)=0$ otherwise. 
For a class $\+C$ of queries, we denote by $\intro*\monDRShapley{\+C}$ the task of given $q \in \+C$ a database $\+D$ and a fact $\alpha \in \+D$ computing $\Sh(\D,\monDRscorefun[q],\alpha)$.\footnote{The `$\mondrasticindex$' superscript stands for \ul{p}ositive support \ul{dr}astic value.}
\AP
Analogously, the ""MPS-Shapley"" value of a fact $\alpha \in \D$ is $\Sh(D,\monMSscorefun[q],\alpha)$ where
 $\intro*\monMSscorefun[q](B)$ is the number of "minimal" "positive supports" of $q$ inside $S$, 
 and we denote by $\intro*\monMSShapley{\+C}$ the task of computing it for a class $\+C$ of queries.\footnote{The `$\minmonsupindex$' superscript stands for \ul{m}inimal \ul{p}ositive \ul{s}upports.}

\subsection{Computing MPS-Shapley Values for Positive Facts}
\label{sec:global-wsms}
We start our study of the complexity of computing responsibility scores for positive facts 
by considering the task $\monMSShapley{}$. %
First we remark that %
 $\monMSscorefun[q]$  can be equivalently defined via a sum of inverse support sizes:
\begin{lemmarep}\label{lem:frac-equation-for-positive-MS}
$\Sh(\D,\monMSscorefun[q],\alpha)$ is equal to the sum, over all "minimal" "positive supports" $S$ of $q$ in $\D$ containing $\alpha$, of $\frac{1}{|S|}$.
\end{lemmarep}
\begin{proof}
   This proof is analogous to that of \cite[Proposition 4.3]{ourpods25}.
   Consider some "minimal" "positive support" $S$, and let $q_S$ be the query whose only "minimal" positive support is $S$. 
   By the so-called ``Shapley axioms'', which $\Sh(\D,\monMSscorefun[{q_S}],\alpha)$ verifies by definition \cite{shapley1953value}, we have by axiom (Null) that $\Sh(\D,\monMSscorefun[q_S],\alpha) = 0$ for any fact $\alpha \in \D \setminus S$. By axiom (Sym), the remaining values must all have the same value, and their sum is fixed to be $\monMSscorefun[q_S](\D)$, that is, the number of "minimal" "positive supports" in $\D$, which is 1. This means that $\Sh(\D,\monMSscorefun[q_S],\alpha) = \frac{1}{|S|}$ if $\alpha \in S$ or $0$ otherwise. Finally, by axiom (Lin) we can sum the contributions of all positive minimal supports to obtain the desired formula.
\end{proof}

\subparagraph*{Data Complexity}

Regarding data complexity, as was the case for the signed facts approach, we can establish tractable data complexity
for all  "UCQneg"s. 

\begin{proposition}\label{prop:UCQneg-poly-MS-data}
    $\monMSShapley{\text{"UCQneg"}}$ is in polynomial time in data complexity.
\end{proposition}
\begin{proof}
    For every fixed "UCQneg" query $q$ and "database" $\D$, observe that: (1) the size of a "minimal" "positive support" is bounded by the maximum size $N$ of a "CQneg" in $q$, and (2) testing if a set $S$ is a "positive support" is in polynomial time (indeed, it amounts to evaluating $\querypm$ over the database $\spos S \cup \sneg\Dneg$).
    Hence, given $\alpha \in \D$, %
    we can enumerate, in polynomial time in the size of $\D$, all polynomially-many %
    subsets $S$ of $\D$ of size $\leq N$ that contain $\alpha$ and check, for each such subset, whether it is a "minimal" "positive support", adding $\frac{1}{|S|}$ to the current value whenever the check succeeds. %
    The correctness of the result is assured by \Cref{lem:frac-equation-for-positive-MS}.
\end{proof}

The preceding result directly yields tractability of deciding "positive-relevance". As in the case of "signed-relevance", 
the upper bound can be improved to "AC0". 

\ifcameraready
   \begin{proposition}
   \label{prop:pos-relevance-AC0}
      It can be decided in "AC0" "data complexity" whether a "fact" $\alpha \in \D$ is "positive-relevant" "wrt" a given "UCQneg" and "database" $\D$. 
   \end{proposition}
\else
   \begin{propositionrep}
      \label{prop:pos-relevance-AC0}
      It can be decided in "AC0" "data complexity" whether a "fact" $\alpha \in \D$ is "positive-relevant" "wrt" a given "UCQneg" and "database" $\D$. 
   \end{propositionrep}
   \begin{proof}
      The proof is essentially the same as that of \Cref{prop:signed-relevance-AC0}. The condition is now:
      \begin{enumerate}[A.]
         \item \label{prop:pos-relevance-AC0:1} 
         There is a "positive support" (for some disjunct $p \in \querypm$)...
         \item \label{prop:pos-relevance-AC0:2} 
         ...that contains $\alpha$, such that...
         \item \label{prop:pos-relevance-AC0:3} 
         ...it does not strictly contain a "positive support" (for some disjunct $p' \in \querypm$).
      \end{enumerate}
      It can be expressed via the first-order sentence:
      \[
      \underbrace{\bigvee_{p \in \querypm} \bigvee_{\mu: \vars(p) \to \vars(q)}  \Big( \psubmu p \mu}_{\text{\Cref{prop:pos-relevance-AC0:1}}} {} \land {}
      \underbrace{\bigvee_{\beta \in \Atsubpplus p} \alpha = \mu(\beta)}_{\text{\Cref{prop:pos-relevance-AC0:2}}}
      {} \land {}
      \underbrace{\lnot \bigvee_{p' \in \querypm} \bigvee_{\mu' : \vars(p') \to \vars(q) \text{ s.t. } \mu'(\Atsubpplus{p'}) \subsetneq \mu(\Atsubpplus p)} \psubmu{p'}{\mu'} \Big)}_{\text{\Cref{prop:pos-relevance-AC0:3}}}
      \]
      where $\AP\intro*\Atsubpplus p$ denotes the set of \emph{positive} relational atoms of the form $\spos R(\bar t)$ in $p \in \querypm$.
   \end{proof}
\fi

\subparagraph*{Combined Complexity}
We cannot proceed by reduction to the monotone case as was done in \Cref{prop:signed-ms} for computing $\monMSShapley{}$. This is because while the 
"minimal" "signed support"s of $q$ in $\D$ coincide with the "minimal support"s of $\querypm$ in $\Dpm$, in general there is no bijection between the "positive support"s of $q$ in $\D$ 
and the latter supports. Moreover, we need to know the number of "positive support"s per size, and
 even in cases where a "positive support" $S$ is induced from a unique "minimal" "(signed) support@signed support" $S'$, %
 the sizes of $S$ and $S'$ will differ as soon as $S'$ contains "negative facts" of the `$\sneg P$' relations. %
However, when the query has no "mergeable atoms", the "minimal" "signed supports" are in bijection with the "minimal" "positive supports" (\Cref{lem:no-mergeable-atoms-implies-bijection}). We exploit this property
 to establish the following tractability result:

\ifcameraready
   \begin{proposition}\label{prop:gmsbarity}
      For any class $\+C$ of "CQneg"s having bounded "generalized hypertree width", no "mergeable atoms", and "bounded negative arity", $\monMSShapley{\+C}$ is in polynomial time.
   \end{proposition}
   \begin{proofsketch}
      Consider a class $\+C$ of "CQneg"s satisfying the stated conditions, and take some input instance $q,\D,\alpha$.
      The "minimal" "positive supports" are in bijection with the "minimal" "signed supports" of the previous section by \Cref{lem:no-mergeable-atoms-implies-bijection}. In particular, we can transform every "minimal" "signed support" into a "minimal" "positive support" by simply removing the "negative facts".
      Leveraging on this bijection, it suffices to build $\Dpmq$ by \Cref{lem:Dpm-restriction-to-negative-atoms}. We then have that the "minimal" "positive supports" of $q$ in $\D$ are in bijection with %
      the "minimal supports" for $\querypm$ in $\Dpmq$. However, we need to count them for every size and the bijection does not preserve sizes (since it adds/removes "negative facts"). This issue is fairly minor, but it nevertheless requires delving into the proof of \cite[Theorem 6.6]{ourpods25}, see appendix for details.
   \end{proofsketch}
\else
   \begin{propositionrep}\label{prop:gmsbarity}
      For any class $\+C$ of "CQneg"s having bounded "generalized hypertree width", no "mergeable atoms", and "bounded negative arity", $\monMSShapley{\+C}$ is in polynomial time.
   \end{propositionrep}
   \begin{proofsketch}
      Consider a class $\+C$ of "CQneg"s satisfying the stated conditions, and take some input instance $q,\D,\alpha$.
      The "minimal" "positive supports" are in bijection with the "minimal" "signed supports" of the previous section by \Cref{lem:no-mergeable-atoms-implies-bijection}. In particular, we can transform every "minimal" "signed support" into a "minimal" "positive support" by simply removing the "negative facts".
      Leveraging on this bijection, it suffices to build $\Dpmq$ by \Cref{lem:Dpm-restriction-to-negative-atoms}. We then have that the "minimal" "positive supports" of $q$ in $\D$ are in bijection with %
      the "minimal supports" for $\querypm$ in $\Dpmq$. However, we need to count them for every size and the bijection does not preserve sizes (since it adds/removes "negative facts"). This issue is fairly minor, but it nevertheless requires delving into the proof of \cite[Theorem 6.6]{ourpods25}, see appendix for details.
   \end{proofsketch}
   \begin{proof}
      Consider a class $\+C$ of "CQneg"s satisfying the stated conditions, and take some input instance $q,\D,\alpha$.
      The "minimal" "positive supports" are in bijection with the "minimal" "signed supports" of Section \ref{sec:signed} by \Cref{lem:no-mergeable-atoms-implies-bijection}. In particular, we can transform every "minimal" "signed support" into a "minimal" "positive support" by simply removing the "negative facts".
      Leveraging on this bijection, it suffices to build $\Dpmq$ by \Cref{lem:Dpm-restriction-to-negative-atoms}. We then have that the "minimal" "positive supports" of $q$ in $\D$ are in bijection with %
      the "minimal supports" for $\querypm$ in $\Dpmq$. However, we need to count them for every size and the bijection does not preserve sizes (since it adds/removes "negative facts"). This issue is fairly minor, but it nevertheless requires delving into the proof of \cite[Theorem 6.6]{ourpods25}.

      Intuitively, the algorithm in \cite[Theorem 6.6]{ourpods25} builds a collection $\mathrm{P}_q$ of queries, designed so that every "minimal support" for $q$ in $\D$ is the image of one, and only one, $q_E \in \mathrm{P}_q$ \cite[Claim E.11]{ourpods25arxiv}, and that all "minimal supports" associated with a given $q_E$ have the same size (the number of relational atoms in $q_E$ as it turns out) \cite[Claim E.8]{ourpods25arxiv}. %
      The algorithm then proceeds to count the "minimal supports" associated with each $q_E \in \mathrm{P}_q$ that contain the input fact $\alpha$ and then ``labels'' that count with the correct size, so they can finally be all aggregated. %

      Because of the bijection, everything in the above paragraph still applies when replacing "minimal supports" (for $\querypm$ in $\Dpmq$) with "minimal" "positive supports" (for $q$ and $\D$) except the common size of all "minimal" "positive supports" associated with a given $q_E$ is now the number of \emph{positive} atoms in $q_E$. Therefore, the only thing we need to change in the algorithm is to replace the `label' applied to the count of $q_E$ from the total number of atoms to the number of positive atoms only.
   \end{proof}
\fi

We further show that we can drop the restriction to queries of "bounded negative arity" %
 by instead imposing that negations are `guarded'.
\AP
A negation $\lnot R(\bar x)$ in a "CQneg" $q$ is ""guarded@guarded negation"" %
 if there exists a positive atom $S(\bar y)$ of $q$ such that $\bar y$ contains all variables of $\bar x$. A "CQneg" is "guarded@guarded negation" if all negations are "guarded@guarded negation", and we denote by ""gCQneg"" the class of all "guarded@guarded negation" "CQneg"s.
\begin{propositionrep}\label{prop:gmsguard}
   For any class $\+C$ of "gCQneg"s having bounded "generalized hypertree width" and no "mergeable atoms", $\monMSShapley{\+C} \in \FP$.
\end{propositionrep}
\begin{proof} 
   Consider such a class $\+C$, and an input instance $q,\D, \alpha$.
   Let us divide $q$ into its non-"negated@negated atom" and "negated atoms" $q = q_+ \land q_-$, where we put "inequality atoms" in $q_-$.
   \AP
   Remember that, by "guardedness@guarded negation", for every "negated atom" $\lnot\gamma$ of $q_-$ there is a "positive atom" $\intro*\myCheck\gamma$ of $q_+$ such that $\vars(\gamma) \subseteq \vars(\myCheck \gamma)$. 
   Further, by "non-mergeability@mergeable atom", for every fact $\beta$ of $\D$, there is at most one atom $\intro*\myHat\beta$ of $q_+$ such that $\myHat\beta \homto \beta$.
   Let $\D' \subseteq \D$ be the set of all facts $\beta \in \D$ such that $\myHat \beta \homto \beta$ for some $\myHat\beta \in q_+$, and observe that any "minimal" "positive support" must be inside $\D'$.
   Note that $\D'$ can be computed in polynomial time since it amounts to testing whether a fact is in the result of an atomic "CQ".
   Let $\D''\subseteq \D'$ be the result of further removing any fact $\beta$ from $\D'$ if it is not part of the evaluation of $q_\beta = \myHat \beta \land \bigwedge_{\myCheck \gamma = \myHat \beta} \lnot\gamma$ on $\D$.
   Notice that $\D''$ can be computed in polynomial time since each $q_\beta$ is polynomial-time tractable.%
   \ifcameraready
      \ We can further show that a set $S\subseteq \D$ is a "minimal" "positive support" of $q$ in $\D$ if, and only if, it is a minimal support of $q_+$ on $\D''$ \cite[Claim~C.1]{thispaper-arxiv}.
   \else
      \begin{claim}
         A set $S\subseteq \D$ is a "minimal" "positive support" of $q$ in $\D$ if, and only if, it is a minimal support of $q_+$ on $\D''$.
      \end{claim}
      \begin{nestedproof}
         From left to right, if $S$ is a "minimal" "positive support", let us show that (1) $S \subseteq \D''$, (2) $S$ is a "(positive) support@support" of $q_+$, and (3) $S$ is a "minimal" support of %
         $q_+$.

         \begin{enumerate}[(1)]
            \item Consider any "satisfying assignment" $\nu : \vars(q) \to \const(\D)$ such that $\nu(q_+) = S$. It is clear that all atoms of $S$ must be in $\D'$ by definition. Further, if there was one atom $\beta \in S$ which is not in $\D''$, it must be because $\beta \not\in q_\beta(\D)$. Since $\myHat \beta$ is the only positive atom of $q$ that can be mapped to $\beta$, we would obtain that $\nu(\myHat \beta) = \beta$ and further that $\nu(\gamma) \not\in \D$ for every $\gamma$ such that $\myCheck \gamma = \myHat \beta$ (since $\nu$ is a "satisfying assignment"). Hence, by definition of $q_\beta$, $\beta$ must belong to $q_\beta(\D)$, a contradiction with the previous statement.
            \item As $S$ is a "positive support" of $q$, we have $\spos S \cup \sneg \Dneg[\D] \models \querypm$, which means that $S$ satisfies all positive atoms in $q$, "ie", $S \models q_+$.
            \item By means of contradiction, if $S' \models q_+$ for some $S' \subsetneq S$, consider $h : q_+ \homto S'$, and let us show that $h$ is a "satisfying assignment" 
            for $\querypm$ on $\spos S' \cup \sneg \Dneg$ contradicting that $S$ is a "minimal" "positive support".
            Since $\D'' \subseteq \D$, it suffices to show that $h(\gamma) \not\in \D$ for every $\lnot\gamma$ in $q_-$.
            Take any such $\lnot\gamma$, and let $\beta = h(\myCheck \gamma) \in \D''$. 
            By definition of $\D''$, we must have that $\beta \in q_\beta(\D)$, and in particular this means that $h(\gamma) \not\in\D$ since $\vars(\gamma) \subseteq \vars(\myCheck \gamma)$.
         \end{enumerate}

         From right to left, if $S$ is a "minimal support" of $q_+$ in $\D''$, let us show that (1) $S \subseteq \D$, (2) $S$ is a "positive support" in $\D$ for $q$, and (3) $S$ is a "minimal" "positive support" for $q$  in $\D$.
         \begin{enumerate}[(1)]
            \item Trivial since $S \subseteq \D'' \subseteq \D$.
            \item Let $h: q_+ \homto S$. Since $h(\myCheck \gamma) \in \D''$ for every $\lnot\gamma$ of $q_-$, this means that $h(\myCheck \gamma) \in q_{h(\myCheck \gamma)} (\D)$ by definition of $\D''$ and thus $h(\gamma) \not\in \D$.
            Hence, $h$ is a "satisfying assignment" for $\querypm$ on $\spos S \cup \sneg \Dneg[\D'']$.
            \item If there was $S' \subsetneq S$ such that $S' \models q_+$, by the argument of the preceding item we would obtain $\spos S' \cup \sneg \Dneg[\D''] \models \querypm$.\qedhere
         \end{enumerate}
      \end{nestedproof}
      \smallskip
   \fi
   Hence, \ifcameraready\else by the previous claim,\fi \ there is a polynomial-time reduction from $\monMSShapley{\+C}$ to  $\msShapley{\+C_+}$, for $\+C_+ = \set{q_+ : q \in \+C}$. Since the queries of $\+C_+$ have no "mergeable atoms" and the same "generalized hypertree width" bound as $\+C$, we can apply \Cref{Theorem 6.6 and more from PODS25}
   to conclude that $\msShapley{\+C_+}$, and thus $\monMSShapley{\+C}$, are in polynomial time.
\end{proof}

\subsection{Computing Drastic-Shapley Values for Positive Facts}
\label{sec:global-drshapley}
\AP
The only existing tractability result for computing a drastic-Shapley value for "CQneg"s concerns the class of "self-join free" queries \cite{ReshefKL20}, and so we shall also restrict our attention to the class of "sjfCQneg"s. To transfer this tractability result to our setting, we rely on a slight adaptation of the notion of ``non-hierarchical paths'' introduced in \cite[§4.1]{ReshefKL20}. 
The definition makes use of the ""Gaifman graph"" $\intro*\Gaifman q$ of a query $q$, whose set of vertices is $\vars(q)$ and which contains an edge between $u$ and $v$ if some atom of $q$ contains both $u$ and $v$.
We say a "CQneg" $q$ has a ""non-hierarchical neg-path""\footnote{This definition coincides with that of ``non-hierarchical path'' from \cite{{ReshefKL20}} except that in Item (1), we forbid relations in "negated atoms", rather than excluding atoms using exogenous relations.} if there are two atoms $\alpha_x$, $\alpha_y$ whose relations are $R_x$ and $R_y$ and two variables $x,y$ such that:
(1) the relations $R_x$ and 
$R_y$ do not occur in "negated atoms", 
(2) variable $x$ appears in $\alpha_x$ but not $\alpha_y$, while variable $y$ appears  in $\alpha_y$ but not $\alpha_x$, 
and (3) the graph obtained from $\Gaifman q$ by removing every vertex corresponding to a variable in $\alpha_x$ or $\alpha_y$ (excepting $x$ and $y$) contains a path between $x$ and $y$.

\ifcameraready
   \begin{proposition}\label{prop:pos-dr-tractability}
      If a "sjfCQneg" $q$ has no "non-hierarchical neg-path", then $\monDRShapley{q} \in \FP$.
   \end{proposition}
\else
   \begin{propositionrep}\label{prop:pos-dr-tractability}
      If a "sjfCQneg" $q$ has no "non-hierarchical neg-path", then $\monDRShapley{q} \in \FP$.
   \end{propositionrep}
   \begin{proof}
   We %
   establish tractability by reducing the considered problem to a tractable setting for a different drastic-Shapley-based responsibility measure from \cite{{ReshefKL20}}. In the latter work, it is assumed that the input database $\D$ is partitioned into an exogenous part $\D_\mathsf{x}$
   and endogenous part $\D_\mathsf{n}$, with only the endogenous facts receiving values. Furthermore, one may specify a subset $X$ of the relations as being purely exogenous, "ie", these relations can only contain exogenous facts (the other relations may have both exogenous and endogenous facts). 
   The task studied in \cite[§4]{{ReshefKL20}} is to compute, for a fixed "CQneg" $q$, "schema" $\Sigma$, and set $X$ of exogenous relations, the value $\Sh(\D_\mathsf{n}, \xi^\mathsf{dr, x}_{q,\D}, \alpha)$, where the input is a partitioned database $\D=(\D_\mathsf{x}, \D_\mathsf{n})$ and fact $\alpha \in \D_\mathsf{n}$. The "wealth function" $\xi^\mathsf{dr, x}_{q,\D}: 2^{\D_\mathsf{n}} \rightarrow  \Reals$ used in %
   \cite{{ReshefKL20}} is as follows:  $\xi^\mathsf{dr, x}_{q,\D}(B) = 1$ if $B \cup  \D_\mathsf{x} \models q$ and $\D_\mathsf{x} \not \models q$, 
   $\xi^\mathsf{dr, x}_{q,\D}(B) = -1$ if $B \cup  \D_\mathsf{x} \not \models q$ and $\D_\mathsf{x} \models q$, 
   else $\xi^\mathsf{dr, x}_{q,\D}(B) =0$. 

   Let us consider a (Boolean) "sjfCQneg" $q\defeq \exists \bar{x} ~ \bigwedge_{i=1}^{p} P_i(\bar{t_i}) \land \bigwedge_{j=1}^{n} \lnot N_j(\bar{t_j}) \wedge \bigwedge_{\ell=1}^m t_{k_\ell}^1 \neq t_{k_\ell}^2$ over the "schema" $\Sigma=\{P_1, \ldots, P_n, N_1, \ldots, N_n\}$.
   As $q$ is "self-join free", we further know that every $P_i$ and $N_j$ is a distinct relation.  
   In order to reuse results from \cite{{ReshefKL20}}, we will modify the "schema" and query as follows. We take the "schema" $\Sigma' = \{P_1,\ldots, P_p, \bar{N}_1, \ldots, \bar{N}_n, \bar{E}_1, \ldots, \bar{E}_k  \}$, where intuitively $\bar{N}_j$ will contain the complement of $N_j$ and each $\bar{E}_k$ will contain a separate copy of the binary inequality relation. 
   Translating the original query into the new "schema", we obtain:
   $$q'\defeq \exists \bar{x} ~ \bigwedge_{i=1}^{p} P_i(\bar{t_i}) \land \bigwedge_{j=1}^{n} \bar{N}_j(\bar{t_j}) \wedge \bigwedge_{\ell=1}^m \bar{E}_k(x_{k_\ell}^1, x_{k_\ell}^2)$$ We can also translate any $\Sigma$-database $\D$ into a $\Sigma'$-database in the expected way:
   \begin{align*}
   \D' = & \{P_i(\bar{c}) : 1 \leq i \leq p, P_i(\bar{c}) \in \D\} \cup \{\bar{N}_j(\bar{c}) : 1 \leq j \leq n, \const(\bar c) \subseteq \const(\D), N_j(\bar{c}) \not \in \D\}\\ & \cup \{\bar{E}_k(c,d) : 1 \leq k \leq \ell, c,d \in \const(\D) \cup \const(q), c \neq d\}
   \end{align*}
   We will fix $X=\{\bar{N}_1, \ldots, \bar{N}_n, \bar{E}_1, \ldots, \bar{E}_k  \}$ as the exogenous relations, and define $\D'_\mathsf{x}$ (resp.\ $\D'_\mathsf{n}$) as those facts in $\D'$ whose relation belongs to (resp.\ does not belong to) $X$. 

   It can be verified that for every $\alpha \in \D$: $\Sh(\D,\xi_{q,\D}^{\mondrasticindex},\alpha)\neq 0$ iff $\alpha \in \D \cap \D' = \D'_\mathsf{n}$ ("ie", has relation $P_i$). Moreover, one can show that $\Sh(\D,\xi_{q,\D}^{\mondrasticindex},\alpha)= \Sh(\D'_\mathsf{n},\xi^\mathsf{dr, x}_{q',\D'},\alpha)$. 
   Intuitively this holds because we ensure that the negated and inequality facts are always present by treating them as exogenous. 
   Thus, it suffices to first construct $q'$ and $\D'$, then 
   compute $\Sh(\D'_\mathsf{n},\xi_{q,\D}^{\mondrasticindex},\alpha)$. Importantly, as $q$ is "self-join free" and has no "non-hierarchical neg-path", the query $q'$ will also be "self-join free" and without any non-hierarchical path, and thus the latter value can be computed in  $\FP$ due to \cite[Theorem 4.3]{{ReshefKL20}}. Membership in $\FP$ then follows from the fact that $q'$, $\Sigma'$, and $X$ can be constructed independently of $\D$, and
   $\D'$ can be constructed in $\FP$ in data complexity from $\D$.
   \end{proof}
\fi

The hardness side of \cite[Theorem 4.3]{ReshefKL20} cannot be used to show the hardness of the remaining "sjfCQneg" with "non-hierarchical neg-path"s because it heavily relies on "exogenous" facts which we do not consider here. However, %
by using relations of "negated atoms" to `encode' exogenous relations we can sometimes transfer intractability results to our setting: %

\ifcameraready
   \begin{proposition}%
   Consider the "CQneg" $q_{\mathsf{R\lnot ST}} \defeq \exists x,y. R(x) \land \lnot S(x,y) \land T(y)$. Then $\mathrm{SVC}^{\mondrasticindex}_{q_{\mathsf{R\lnot ST}}}$ is $\sP$-hard.
   \end{proposition}
\else
   \begin{propositionrep}%
   Consider the "CQneg" $q_{\mathsf{R\lnot ST}} \defeq \exists x,y. R(x) \land \lnot S(x,y) \land T(y)$. Then $\mathrm{SVC}^{\mondrasticindex}_{q_{\mathsf{R\lnot ST}}}$ is $\sP$-hard.
   \end{propositionrep}
   \begin{proof}
   We can observe that in the proof showing $\sP$-hardness of the "CQ" $q_{\mathsf{RST}} \defeq \exists x,y. R(x) \land S(x,y) \land T(y)$ for the drastic-Shapley measure in \cite[Proposition 4.6]{livshitsShapleyValueTuples2021}, all of the considered partitioned databases
   designate all tuples in $R$ and $T$ as "endogenous", and all tuples in $S$ as "exogenous". 
   Given any database $\D$ over the query "schema" $\{R,S,T\}$, let us denote by $\D_{\bar S}$ the database obtained from $\D$ by leaving the $R$ and $T$ relations untouched and including the fact $S(c,d)$ in $\D_{\bar S}$ iff $S(c,d) \not \in \D$ and $c,d \in \const(\D)$ ("ie", we replace $S$ by its complement over the "active domain").
   By examining the definition of $\xi^\mathsf{pdr}$, %
   one can show that for every partitioned database $\D=(\D_\mathsf{n}, \D_\mathsf{x})$ over "schema" $\{R,S,T\}$ such that all $R$- and $T$-facts are in $\D_\mathsf{n}$
   and all $S$-facts in $\D_\mathsf{x}$,
   the following holds for every $\alpha \in \D_\mathsf{n}$: 
   $$\Sh(\D_\mathsf{n}, \xi^\mathsf{dr,x}_{q_{\mathsf{RST}}}, \alpha) = \Sh(\D_{\bar S}, \xi^\mathsf{pdr}_{q_{\mathsf{R\lnot ST}}}, \alpha)$$
   where the `$\mathsf{x}$' in the superscript of $\xi^\mathsf{dr,x}_{q_{\mathsf{RST}}}$ indicates that we are taking into account the exogenous facts (see proof of Proposition \ref{prop:pos-dr-tractability} for details). This follows from the facts that (i) for every $B \subseteq \D_\mathsf{n}$: $\xi^\mathsf{dr,x}_{q_{\mathsf{RST}}}(B) = \xi^\mathsf{pdr}_{q_{\mathsf{R\lnot ST}}}(B)$, and (ii) each "positive $S$-fact@positive fact" $\beta \in \D_{\bar S}$ is not "positive-relevant", hence $\Sh(\D_{\bar S}, \xi^\mathsf{pdr}_{q_{\mathsf{R\lnot ST}}}, \beta)=0$ (so the placement of such elements in orderings of $\D_{\bar S}$ can be ignored). 
   It follows that we may carry out essentially the same reduction as in \cite[Proposition 4.6]{livshitsShapleyValueTuples2021}, 
   but using instead the query $q_{\mathsf{R\lnot ST}}$, the "wealth function" $\xi^\mathsf{pdr}_{q_{\mathsf{R\lnot ST}}}$, and the database $\D_{\bar S}$,
   thereby establishing $\sP$-hardness of $\mathrm{SVC}^{\mondrasticindex}_{q_{\mathsf{R\lnot ST}}}$. 
   \end{proof}
\fi

%% file: sec-relwork-new.tex
\section{Related Work on Shapley Values for Queries with Negation}\ifcameraready\nosectionappendix\fi %
\label{sec:local-monotone}
Another notion of relevance that has been put forth is based on the `impact' that the addition of a given fact may have on the satisfaction of the query on a sub-database. More concretely, for a given Boolean query $q$ and a database $D$, a fact $\alpha$ has a \AP""positive impact"" ("resp" \AP""negative impact"") on $D' \subseteq D$ if we have $D' \not\models q$ and $D' \cup \set\alpha \models q$ ("resp" $D'  \models q$ and $D' \cup \set\alpha \not\models q$). 
A fact is then \AP""impact-relevant"" -- called `relevant' in \cite{ReshefKL20} -- if it has some \AP""impact"" (positive, negative, or both) on some sub-database.%
    \footnote{
        Observe that, for monotone queries, it is equivalent whether a fact is "impact-relevant", has "positive impact", is part of a "minimal support", or is an actual cause in the framework of causal responsibility \cite{DBLP:journals/pvldb/MeliouGMS11}. 
    }

The article \cite{ReshefKL20} introduces and investigates precisely this notion of relevance, in particular studying the complexity of whether a fact is "impact-relevant" for a given query.
Through this lens, the cited work assigns a positive responsibility score to a fact which has a "positive impact", a negative score to those with a "negative impact", and null score to facts which have no "impact".
But it should be noted that facts may have simultaneously positive and negative "impacts", in which case some relative weight of the `positivity' and `negativity' must be computed, which in particular may output a null score, even though the fact has an "impact" and thus is "impact-relevant" under this view.

Akin to what is done in the case of monotone queries, \cite{ReshefKL20} defines a drastic-Shapley value of a fact $\alpha$ in a "database" $\D$ for a (non-monotone) query $q$ as being proportional to the number of sub-databases on which $\alpha$ has a "positive impact" minus the number of sub-databases on which it has a "negative impact". 
More precisely, they consider the value $\Sh(\D, \drscorefun[q], \alpha)$ for the "drastic@drastic-Shapley" wealth function $\drscorefun[q]$  (defined as $\drscorefun[q](S) = 1$ if $S \models q$  and $\drscorefun[q](S) = 0$ otherwise).\footnote{We stress that here we are adapting \cite{ReshefKL20}'s definition to our simpler setting without  "exogenous facts".}
Observe that while a non-null score implies "impact-relevance", the converse does not always hold (this was already observed by the authors \cite[Example~5.3]{ReshefKL20}, see also \Cref{prp:impact-deter}). 
However, if a fact has "positive impact" but no "negative impact", the score will be strictly positive.
In the example of \Cref{fig:ex-recipe}, the fact $\alpha = I(\textit{mp}, \textit{wine})$ has positive score $\Sh(\D, \drscorefun[q], \alpha)>0$ since: (a) it makes a "positive impact" on the empty sub-database, and (b) it cannot ever make a "negative impact" on a sub-database -- in other words, $\drscorefun[q](S) = 0$ implies $\drscorefun[q](S \setminus \set\alpha) = 0$.
The following results show that "impact-relevance" is orthogonal to the relevance notions we have introduced in our work.

\begin{lemmarep}
    There exist  "impact-relevant" facts that are neither "positive-relevant" nor "signed-relevant" for "CQneg" queries, 
    and facts that are "positive-relevant" and "signed-relevant" but not "impact-relevant" for "UCQneg" queries.
\end{lemmarep}
\begin{proof}
    \proofcase{First statement.} In the example of \Cref{fig:ex-recipe}, the fact $I(\textit{mp}, \textit{wine})$ is "impact-relevant" but not "signed-relevant" nor "positive-relevant".

    \smallskip

    \proofcase{Second statement.} Consider the "UCQneg" 
            $q = (\exists x ~  A(x) \land B(x))
                        \lor       
                (\exists x ~ \lnot A(x) \land B(x))$
    and the database $\D = \set{A(c),B(c)}$.
    It is easy to see that the fact $A(c)$ is not "impact-relevant". 
    However, $\spos A(c)$ is "signed-relevant" since it belongs to the "minimal" "signed support" $\set{\spos A(c), \spos B(c)}$, and $A(c)$ is "positive-relevant" since it belongs to the "minimal"   "positive support" $\set{ A(c), B(c)}$.
\end{proof}

\ifcameraready
    \begin{proposition}\label{prp:impact-deter}
    There exist a database $\D$, a "UCQneg" $q$, a "positive-relevant" and "impact-relevant" fact $\alpha \in \D$ and a non-"positive-relevant" fact $\beta \in \D$ such that  
    $\Sh(\D, \drscorefun[q], \alpha) = 0$ and 
    $\Sh(\D, \drscorefun[q], \beta)\neq 0$. Further, $\spos \alpha$ is "signed-relevant" and $\spos\beta$ is not.%
    \footnote{%
    We currently do not have an example of a "positive-relevant" fact which has null-score for a "CQneg" (instead of a "UCQneg"), and hence we do not know if \Cref{prp:impact-deter} holds when replacing "UCQneg" with "CQneg".}
    \end{proposition}
    \begin{proofsketch}
    Building on the example of \Cref{fig:ex-recipe}, with 
    the query $q_2 = \exists x ~ (I(x,\textit{meat}) \land I(x,\textit{wine})) \lor (I(x,\textit{fish}) \land \lnot I(x,\textit{wine}))$, we observe that $\alpha = I(mp,\textit{wine})$ is "positive-@positive-relevant" and "signed-relevant" but $\Sh(\D, \drscorefun[q], \alpha) = 0$ (since the "positive@positive impact" and "negative impacts" cancel out), and $\beta = I(mm,\textit{wine})$ is neither "positive-@positive-relevant" nor "signed-relevant" yet has $\Sh(\D, \drscorefun[q], \beta) < 0$.
    \end{proofsketch}
\else
    \begin{propositionrep}\label{prp:impact-deter}
    There exist a database $\D$, a "UCQneg" $q$, a "positive-relevant" and "impact-relevant" fact $\alpha \in \D$ and a non-"positive-relevant" fact $\beta \in \D$ such that  
    $\Sh(\D, \drscorefun[q], \alpha) = 0$ and 
    $\Sh(\D, \drscorefun[q], \beta)\neq 0$. Further, $\spos \alpha$ is "signed-relevant" and $\spos\beta$ is not.%
    \footnote{%
    We currently do not have an example of a "positive-relevant" fact which has null-score for a "CQneg" (instead of a "UCQneg"), and hence we do not know if \Cref{prp:impact-deter} holds when replacing "UCQneg" with "CQneg".}
    \end{propositionrep}
    \begin{proofsketch}
    Building on the example of \Cref{fig:ex-recipe}, with 
    the query $q_2 = \exists x ~ (I(x,\textit{meat}) \land I(x,\textit{wine})) \lor (I(x,\textit{fish}) \land \lnot I(x,\textit{wine}))$, we observe that $\alpha = I(mp,\textit{wine})$ is "positive-@positive-relevant" and "signed-relevant" but $\Sh(\D, \drscorefun[q], \alpha) = 0$ (since the "positive@positive impact" and "negative impacts" cancel out), and $\beta = I(mm,\textit{wine})$ is neither "positive-@positive-relevant" nor "signed-relevant" yet has $\Sh(\D, \drscorefun[q], \beta) < 0$.
    \end{proofsketch}
    \begin{proof}
    The instance is an adaptation of the example depicted in \Cref{fig:ex-recipe}, with the query $q_2 = \exists x ~ (I(x,\textit{meat}) \land I(x,\textit{wine})) \lor (I(x,\textit{fish}) \land \lnot I(x,\textit{wine}))$.
    Let $\alpha = I(mp,\textit{wine})$ and observe that $\alpha$ is "positive-relevant" since it belongs to the "minimal" "positive support" $\set{I(mp,\textit{meat}),I(mp,\textit{wine})}$ and it is "signed-relevant" since it belongs to the "minimal" "signed support" $\set{\spos I(mp,\textit{meat}),\spos I(mp,\textit{wine})}$. Let us now see the impact that $\alpha$ has on the different orderings.
    For brevity, let us denote the facts $I(mp,\textit{wine})$, $I(mp,\textit{meat})$, $I(mp,\textit{fish})$, $I(mm,\textit{wine})$, $I(mm,\textit{fish})$ of $\D$ using the numbers 1 to 5 (hence, under this nomenclature, $\alpha$ is $1$), and let us consider all linear orderings on $\set{1,2,3,4,5}$. The following are the orderings having a "positive impact" on $\alpha$ ("ie", those where the sub-database of facts before the position where $\alpha$ sits is not a "support" of $q_2$ but when adding $\alpha$ it becomes a "support"): 
    21345, 
    21354, 
    21435, 
    21453, 
    21534, 
    21543, 
    24135, 
    24153, 
    25413, 
    42135, 
    42153, 
    52413, 
    54213. 
    And these are the ones having a "negative impact": 
    31245, 
    31254, 
    31425, 
    31452, 
    31524, 
    31542, 
    34125, 
    34152, 
    35412, 
    43125, 
    43152, 
    53412, 
    54312.
    Observe that these are exactly the same number,\footnote{That is, 13 each one, out of the $5!$ orderings. The remaining $94$ ($=5!-2 {\cdot} 13$) orderings  yield no "impact".} and hence that 
    $\Sh(\D, \drscorefun[q], \alpha) = 0$.

    On the other hand, $\beta = I(mm,\textit{wine})$ cannot be in any "minimal" "positive support": indeed, any "positive support" $S$ containing $\beta$ is not "minimal" since $S$ must also contain $\set{I(mp,\textit{meat}),I(mp,\textit{wine})}$ to satisfy $q$ and thus removing $\beta$ still results in a "positive support". 
    It cannot be that $\spos \beta$ is in a "minimal" "signed support" either, for similar reasons. Hence  $(\spos) \beta$ is neither "positive-relevant" nor "signed-relevant".
    Further, $\beta$ is "negatively impactful": indeed adding $\beta$ to a sub-database cannot change the truth value from false to true, but it can change it from true to false. Therefore, 
    $\Sh(\D, \drscorefun[q], \beta) < 0$.

    It is immediate to see from the definitions that the score will still be $0$ for $\alpha$ and still be strictly negative for $\beta$ if we replace $\Sh$ with the Banzhaf Power Index, or any `Shapley-like' score (as defined in \cite{karmakarExpectedShapleyLikeScores2024}).
    \end{proof}
\fi

In \cite[Theorems 3.1 and 4.3]{ReshefKL20} it was shown that computing $\Sh(\D, \drscorefun[q], \alpha)$  for "self-join free" "CQneg" queries $q$ (without inequalities) is tractable in "data complexity" when the query is hierarchical, or more generally, does not contain any non-hierarchical path. As shown in \Cref{prop:pos-dr-tractability,prop:signed-dr-tractability}, analogous tractability results hold for the drastic-Shapley measures  introduced in our work. In contrast, it is shown in \cite[Proposition 5.5]{ReshefKL20} that testing whether a fact is "impact-relevant" is generally an "NP"-complete problem, while we have tractability for testing for "signed-@signed-relevant" and "positive-relevance" ("cf" \Cref{prop:signed-relevance-AC0,prop:pos-relevance-AC0}).

%% file: sec-conclusion.tex
\section{Concluding Remarks}
\label{sec:conclusion}
Our study on responsibility measures led us to investigate and compare different notions of qualitative explanations (supports) and relevance for queries with negation. We believe that these results are of independent interest and note in particular that both "signed-relevance" and "positive-relevance" are easily computable. These notions moreover provide a solid formal underpinning for the novel responsibility measures we introduce, and our complexity study shows how existing tractability results can be smoothly extended to our proposed measures. In particular, the reduction from Lemma \ref{lem:redux:mssigned-to-msnormal}
means that the measures for signed facts can be straightforwardly implemented using algorithms for monotone queries.

While we have focused on developing measures inspired by the "MS-Shapley" and "drastic-Shapley" measures, %
all of the results for signed and positive variants of "MS-Shapley" -- namely \Cref{Theorem 6.6 and more from PODS25,prop:gmsguard,prop:gmsbarity,prop:signed-ms,prop:UCQneg-poly-MS-data} -- can be trivially extended to other "WSMS" measures based on "tractable weight function"s
\ifcameraready%
   \cite[§E.1]{thispaper-arxiv}.
\else%
   (see \Cref{app:wsms} for more details).
\fi

One can naturally extend the current framework to handle `exogenous facts' (cf. \Cref{rk:exogenous}), while preserving the established upper bounds and enabling the possibility of transferring existing intractability results. What is less clear is how to extend our approaches to other classes of queries with negation, and we leave the investigation of meaningful notions of relevance, explanation, and responsibility for other sorts of non-monotone queries -- such as universal queries -- as an interesting but challenging direction for future work. In particular, it would be relevant to explore what are the desirable properties of responsibility measures for non-monotone queries, in line with a recent study for their monotone counterparts \cite{ourpods25}, which could provide a further axiomatic justification for our proposed measures and their generalizations.

\ifcameraready\else
   \begin{toappendix}
   \subsection{Weighted sums of minimal supports}
   \label{app:wsms}
   While we have focused here in the "MS-Shapley", as remarked in the introduction this is part of a larger family of responsibility-attribution scores based on weighted sums of minimal supports, or ""WSMS"". 

   For each \AP""weight function"" $w: \Nat \to \Reals$ we obtain a responsibility score which assigns, to any "database" $\D$ "fact" $\alpha \in \D$ and (monotone) query $q$ the number\footnote{For the sake of simplicity and clarity, we consider the "weight function" $w$ to take only $|S|$ as a parameter, whereas in the definition of \cite[§4.2]{ourpods25} it may also depend on $|D|$.}
   \begin{equation}\label{eq:wsms}
      \phi_{\mathsf{wsms}}^w(\+D, q, \alpha) \defeq \sum_{\substack{S\in\Minsups q(\D)\\\alpha\in S}} w(|S|)
   \end{equation}
   where $\intro*\Minsups q(\D)$ are the "minimal supports" of $q$ in $\D$.
   As it turns out, $\phi_{\mathsf{wsms}}^w$ is also the Shapley value $\Sh(\D, \wscorefun[q]{w}, \alpha)$ for a suitable "wealth function" $\wscorefun[q]{w}$.
   We say that $w$ is a \AP""tractable weight function"" if it can further be computed in polynomial time.

   One can then obtain analogous score functions by replacing $\Minsups q$ with our notions of explanations, such as "minimal" "positive supports". It can still be seen that the resulting $\phi_{\mathsf{wsmps}}^w(\+D, q, \alpha)$ of weighted sums of "minimal" "positive supports" is still a Shapley score, by essentially the same proof as \cite[Proposition 4.4]{ourpods25}.

   Further, the upper-bound results we have seen for "MS-Shapley" --namely \Cref{Theorem 6.6 and more from PODS25,prop:gmsguard,prop:gmsbarity,prop:signed-ms,prop:UCQneg-poly-MS-data}-- can be naturally extended to their "WSMS" analog using any "tractable weight function". This is immediate from the definition  above --as observed in \cite[Proposition~6.1]{ourpods25} in the case of "monotone queries"-- since all tractable cases consist in efficiently computing the number of `minimal explanations' ("ie" "minimal" "positive supports" or "minimal" "signed supports") of any given size.

   \end{toappendix}
\fi